\documentclass[runningheads]{llncs}
\usepackage{amssymb}
\usepackage{graphicx}

\usepackage{url}
\urldef{\mailsa}\path|mpfrank@sandia.gov|    
\newcommand{\keywords}[1]{\par\addvspace\baselineskip
\noindent\keywordname\enspace\ignorespaces#1}

\usepackage{physics}	% Used for ket notation in long version

\spnewtheorem{defn}[definition]{Definition}{\bfseries}{\upshape}
\spnewtheorem{notation}{Notation}{\bfseries}{\upshape}
\spnewtheorem{assertion}[proposition]{Assertion}{\bfseries}{\upshape}
\spnewtheorem{lem}[lemma]{Lemma}{\bfseries}{\upshape}
\spnewtheorem{thm}[theorem]{Theorem}{\bfseries}{\upshape}
\spnewtheorem{cor}[corollary]{Corollary}{\bfseries}{\upshape}

\newcommand{\Delt}{\mathrm\Delta}
\newcommand{\lanprinc}{Landauer's Principle}		% To keep capitalization consistent
\newcommand{\eg}{\emph{e.g.}}

\usepackage{fancyhdr}% http://ctan.org/pkg/fancyhdr
\fancypagestyle{title}{%
  \setlength{\headheight}{22pt}%
  \fancyhf{}% No header/footer
  \fancyhead[C]{\textsc{Extended preprint of ``Foundations of Generalized
  	Reversible Computing,'' 9th Conf.\ on Reversible Computation (RC17).}}
}%

\begin{document}

\mainmatter

\title{Generalized Reversible Computing}

\titlerunning{Generalized Reversible Computing}

\author{Michael P. Frank%
\thanks{This work was supported by the Laboratory Directed Research 
and Development pro\-gram at Sandia National Laboratories and by the 
Advanced Simulation and Computing program under the U.S. Department 
of Energy's National Nuclear Security Administration (NNSA).  Sandia 
National Laboratories is a multimission laboratory managed and 
operated by National Technology and Engineering Solutions of Sandia, 
LLC., a wholly owned subsidiary of Honeywell International, Inc., for 
NNSA under con\-tract DE-NA0003525.  Approved for public release 
SAND2018-6889 O.}%
}
\authorrunning{Michael P.\ Frank}

\institute{Center for Computing Research, Sandia National Laboratories,\\
1515 Eubank SE, Mail Stop 1322, Albuquerque, NM 87123\\
\mailsa\\
\url{http://www.cs.sandia.gov}}

\toctitle{Generalized Reversible Computing}
\tocauthor{Michael P.\ Frank}
\maketitle

\thispagestyle{title}		% Uncomment this line for preprints

%---------1---------2---------3---------4---------5---------6---------7
\begin{abstract}
{\lanprinc} that the loss of information from a computation 
corresponds to an increase in entropy can be expressed as a rigorous 
theorem of mathematical physics.  However, carefully examining its 
detailed formulation reveals that the traditional definition 
identifying logically reversible computational operations with 
bijective transformations of the full digital state space is actually 
not the most general characterization, at the logical level, of the 
complete set of classical computational operations that can be carried 
out physically with asymptotically zero energy dissipation.  To derive 
the correct set of necessary logical conditions for physical 
reversibility, we must take into account the effect of initial-state 
probabilities when applying the detailed form of the Principle.  The 
minimal logical-level requirement for the physical reversibility of 
deterministic computational operations turns out to be that only the 
\emph{subset} of initial states that are assigned nonzero probability 
in a given statistical operating context must be transformed one-to-one 
into final states.  Consequently, \emph{any} computational operation 
can be seen as \emph{conditionally reversible}, relative to any 
sufficiently-restrictive precondition on its initial state, and the 
minimum average dissipation required for any deterministic operation 
by {\lanprinc} asymptotically approaches zero in contexts where the 
probability of meeting any preselected one of its suitable 
preconditions approaches unity.  The concept of conditional 
reversibility facilitates much simpler designs for asymptotically 
thermodynamically reversible computational devices and circuits, 
compared to designs that are restricted to using only fully-bijective 
operations such as Fredkin/Toffoli type operations.  Thus, this more 
general framework for reversible computing provides a more effective 
theoretical foundation to use for the design of practical reversible 
computing hardware than does the more restrictive traditional model 
of reversible logic.  In this paper, we formally develop the 			% "formally develop" = long vers.
theoretical foundations of the generalized model, and briefly survey 
some of its applications.
%---------1---------2---------3---------4---------5---------6---------7
\end{abstract}

%---------1---------2---------3---------4---------5---------6---------7
\section{Introduction}
\label{sec:intro}

As the end of the semiconductor roadmap approaches, there is today a
growing realization among industry leaders, researchers, funding 
agencies and investors that a transition to novel computing paradigms 
will be required in order for engineers to continue improving the 
energy efficiency (and thus, cost efficiency) of computing technology 
beyond the expected final CMOS node, when minimal transistor gate 
energies are expected to plateau at around the 40-80\,$k_{\rm B}T$ 
level\footnote{Where $k_{\rm B}$ is Boltzmann's constant, and $T$ is 
operating temperature.} ($\sim 1$-2 eV at room temperature), with 
typical total $CV^2$ node energies\footnote{Where $C$ is node 
capacitance, and $V$ is logic swing voltage.} plateauing at a much 
higher level of around 1-2 keV \cite{itrs}.  But, no matter at what 		% [1]
level exactly signal energies finally flatten out, to recover and 
reuse a fraction of the signal energy approaching 100\% is going to 
require carrying out logically reversible transformations of the local 
digital state, due to {\lanprinc} \cite{land}, which tells us that 			% [2]
performing computational operations that are irreversible (\emph{i.e.}, 
that lose information) necessarily results in an increase in entropy, 
and thus energy dissipation.  Therefore, it will be essential for the 
engineers who design and use future digital bit-devices to understand 
clearly and precisely what the meaning of and rationale for {\lanprinc} 
really are, and precisely what are the minimal requirements, at the 
logical level, for computational operations to be reversible---meaning, 
both not information-losing, and also capable of being physically 
carried out in an asymptotically thermodynamically reversible way.

However, the definition of logical reversibility that has been 
in wide\-spread use ever since Landauer's original paper is not, in 
fact, the most general defi\-ni\-tion of logical reversibility that is
consistent with the understanding that a lo\-gically reversible 
computational process can, in principle, be carried out via an 
(asymptotically) thermodynamically reversible physical process.  Thus, 
the traditional definition of logical reversibility in fact 
\emph{obscures} most of the space of technological possibilities, and 
has resulted in a substantial amount of debate and confusion 
(\emph{e.g.}, \cite{natcomm}) regarding the issue of whether logical 		% [3]
reversibility is really required for physical 
reversibility.\footnote{In the language of this paper, the authors of 
\cite{natcomm} empirically demonstrate that certain 							% [3*]
conditionally-reversible operations, which we would refer to as 
\texttt{rOR}/\texttt{rNOR}, can avoid the Landauer limit, but without 
realizing that what they are doing is still a form of logically 
reversible computing, in the generalized sense developed here.}  
The debate can be resolved, and the confusion cleared up, by 
understanding that \emph{yes}, logical reversibility (if the meaning 
of that phrase is defined correctly) is \emph{indeed} required for 
physical reversibility, \emph{but}, the \emph{traditional} definition 
of what logical reversibility \emph{means} is actually \emph{not} the 
correct definition for this purpose; it is overly restrictive.  This 
can be rigorously proven by directly applying the detailed 
mathematical formulation of {\lanprinc}.  It turns out that a much 
larger set of computational operations is, in fact, reversible, at the 
logical level, than the traditional definition of ``logical 
reversibility'' acknowledges, and this larger set opens up many 
opportunities for device and circuit engineering that could not have 
been modeled at all by solely using the traditional definition of 
logical reversibility.  Nevertheless, those design opportunities have 
been noticed anyway by many of the engineers (\emph{e.g.}, 
\cite{likh,fred,drexler,yk}) who have developed concepts for hardware 		% [4-7]
implementations of reversible computing.  But, there remains today a 
widespread disconnect between the bulk of reversible computing theory, 
versus the engineering principles required for the design of efficient 
reversible hardware, a disconnect which can be bridged if theorists 
come to understand the necessity of exploring a more general 
theoretical model for reversible computing.  It is the goal of this 
paper to develop such a model from first principles, and show exactly 
why it is necessary and useful.

The structure of the rest of this paper is as follows.  In Section 2, 
we review some physical foundations, and then present a simple, 
general formulation of {\lanprinc} which follows from basic facts of 
mathematical physics and information theory.  This formulation both 
illustrates why {\lanprinc} itself is rigorously true (not debatable), 
and serves as a starting point for later analysis.  Then in Section 3, 
we reformulate the foundations of reversible computing theory in a way 
that develops a new theoretical framework that we call Generalized 
Reversible Computing (GRC), which includes the essential but 
usually-overlooked concept of \emph{conditional reversibility} 
\cite{ismvl}, which generalizes and subsumes the old definition of 			% [8]
(\emph{unconditionally}) logically reversible operations in a way that, 
critically, accounts for the statistical characteristics that apply in 
the context of specific computations.  In Section 4, we present 
several examples of conditionally-reversible operations that are 
useful as building blocks for reversible hardware design, and that are 
also straightforwardly physically implementable.  Many of these 
operations have already been implicitly utilized by the designers of 
various reversible hardware concepts (\eg, \cite{likh,fred,drexler,yk}), 	% [4-7*]
despite the fact that most of the existing reversible computing theory 
literature is completely silent about them, as well as about all other 
operations in the largest, most diverse class of reversible operations, 
those that are not also unconditionally reversible.  Section 5 briefly 
discusses why it is GRC, and not the traditional 
unconditionally-reversible model of reversible computing, that is the 
appropriate model for understanding asymptotically thermodynamically 
reversible hardware such as adiabatic switching circuits.  Section 6 
contrasts GRC's concept of conditional reversibility with existing 
concepts of conditions for correctness of reversible computations that 
have been explored in the literature.  Section 7 concludes, and 
outlines some directions for future work.

A shortened version of this paper titled ``Foundations of Generalized 
Reversible Computing,'' which omitted the proofs, was presented at the 
$9^{\mathrm{th}}$ Conference on Reversible Computation (RC17) in 
Kolkata, India {\cite{FoGRC}}.  A preprint of that shorter version which 	% [9]
included the proofs in an appendix was subsequently posted online
at {\cite{FoGRC-preprint}}.  The present version includes the proofs
inline, as well as additional figures and discussion.  (The intent is to
further expand it in preparation for journal submission.	)					% [10]

%---------1---------2---------3---------4---------5---------6---------7
\section{Formulating {\lanprinc}}
\label{sec:formland}

{\lanprinc} is, in essence, simply the observation that the loss of 
information from a computation corresponds to an increase in physical 
entropy, which implies a certain associated amount of energy being 
dissipated to the environment in the form of heat.  At one level, this 
statement is just a tautological consequence of the understanding that 
the very \emph{meaning} of physical entropy is, in effect, that part 
of the total information embodied within a given physical system that 
has \emph{already} been permanently lost, in the sense of its having 
been ``scrambled up'' (\emph{i.e.}, randomized), through uncertain or 
chaotic interactions with an unknown environment, to the extent that 
the system cannot, in isolation, be effectively restored to its 
original state via any physical procedure that is practically 
accessible to us.  At this level, there is not much more to 
understand---lost information is entropy, and energy is required to 
expunge that entropy to the environment (as heat).  However, for our 
purposes, it is helpful to also articulate the meaning of (and 
justification for) {\lanprinc} in a more thorough and mathematically 
rigorous way.  This is necessary to (for example) understand exactly 
what \emph{information loss} really means, and under what conditions, 
precisely, information is (or is not) lost in the course of carrying 
out a given computation.  As we will see, such an understanding leads 
to the realization that there is, in fact, a much wider variety of 
computational operations that can avoid information loss and entropy 
emission (and thus are reversible, in those senses) than the 
traditional theory of reversible computing acknowledges.

To explain all of this more formally, we start with some basic 
definitions.  We will not here require or provide a full explication 
of quantum theory, but rather, we will work with a simplified set of 
definitions that is adequate for our purposes.  However, our 
definitions will nevertheless be fully compatible with a more complete 
quantum-mechanical treatment.

%---------1---------2---------3---------4---------5---------6---------7
\subsection{Physical state spaces, bijective dynamics.}
\label{sec:physstate}

We first present a basic concept of a space (set) of physical states, 
and explain what we mean when we say that a physical dynamics on a 
given state space is bijective.  The bijectivity of real physical 
dynamics is the fundamental postulate from which {\lanprinc} derives.

\begin{defn} \label{def:ss} \textbf{\emph{State spaces.}} 
For our purposes, a (physical) \emph{state space} is a set 
$\Sigma={s_1, s_2, ... , s_N}$ of $N\in\mathbb{N}$ entities called 
(physical) \emph{states} that are mutually distinct objects from each 
other (mathematically), and that are also reliably distinguishable 
from each other, physically.
\end{defn}

It is important, for our purposes, that states be reliably physically 
distinguishable from each other, as well as being mathematically 
distinct, since otherwise the mathematical distinction between them 
could not reliably convey any information-theoretic content.  An 
example of two states that are physically distinguishable from each 
other would be any two pure quantum states represented by mutually 
orthogonal (perpendicular) quantum state vectors.  In contrast, an 
example of a pair of states that would be mathematically distinct, 
but not reliably physically distinguishable, would be any pair of 
quantum state vectors spanned by any angle $\theta<90^\circ$.

An acceptable example of a state space, for purposes of our definition, 
would therefore be any orthonormal set of basis vectors for any 
$N$-dimensional Hilbert space.

For simplicity, here we assume that the cardinality $N$ of the state 
space is a finite, natural number.  Countably infinite or transfinite 
(\emph{e.g.}, continuous) state spaces are not explicitly considered 
in this paper; however, this is not a significant limitation, since it 
is believed \cite{lloyd} that the accessible universe has only 				% [11]
finitely many distinguishable states anyway, and, even if that turns out 
to be incorrect, certainly any practically-buildable technology will 
only be able to access finite-sized physical systems exhibiting only 
finitely many distinguishable states for the foreseeable future 
(barring major upheavals in fundamental physics).

Next, we need a concept of a bijective (reversible and deterministic) 
dynamics:

\begin{defn} \label{def:bijdyn} \textbf{\emph{Bijective dynamics.}}
Given any state space $\Sigma$ of possible states of a system at some 
reference time point $t_0 \in \mathbb{R}$, a \emph{deterministic, 
reversible dynamics} or simply \emph{bijective dynamics} on $\Sigma$ 
is a parameterized family of total, one-to-one, single-valued 
functions $D(\Delt t):\Sigma\rightarrow \Sigma(\Delt t)$, where the
parameter $\Delt t\in\mathbb{R}$, mapping states in $\Sigma$ onto 
states in the parameterized family of state spaces $\Sigma(\Delt t)$, 
and where $\Sigma(\Delt t)$ always has the same cardinality as $\Sigma$ 
(that is, $|\Sigma(\Delt t)|=|\Sigma|$ for all $\Delt t\in\mathbb{R}$).  
Additionally, we require that $\Sigma(0)=\Sigma$, and that $D(0)$ 
must be the identity function on $\Sigma$.
\end{defn}

In this definition, the real-valued parameter $\Delt t$ represents 
elapsed time from the reference time point $t_0$, and $D(\Delt t)$ 
is the function mapping states that are in $\Sigma$ at the initial 
time point $t_0$ to the new states in $\Sigma(\Delt t)$ that they will 
become after the time $\Delt t$ has elapsed.  See 
Figure~\ref{fig:bijdyn}.  The set $\Sigma(\Delt t)$ is simply the 
space of possible states at time $t_0 + \Delt t$, given that the state 
at time $t_0$ is in $\Sigma$.  Note that $\Sigma(\Delt t)$ will not, 
in general, be the \emph{same} state space as $\Sigma$ for elapsed 
times $\Delt t\neq 0$, since states may, in general, transform 
continuously over time, yet, the initial state space $\Sigma$ under 
consideration was assumed to only be countable, so it does not itself 
include sufficiently many states to allow continuous change while 
still remaining within the same state set.   

\begin{figure}
\centering
\includegraphics[height=2in]{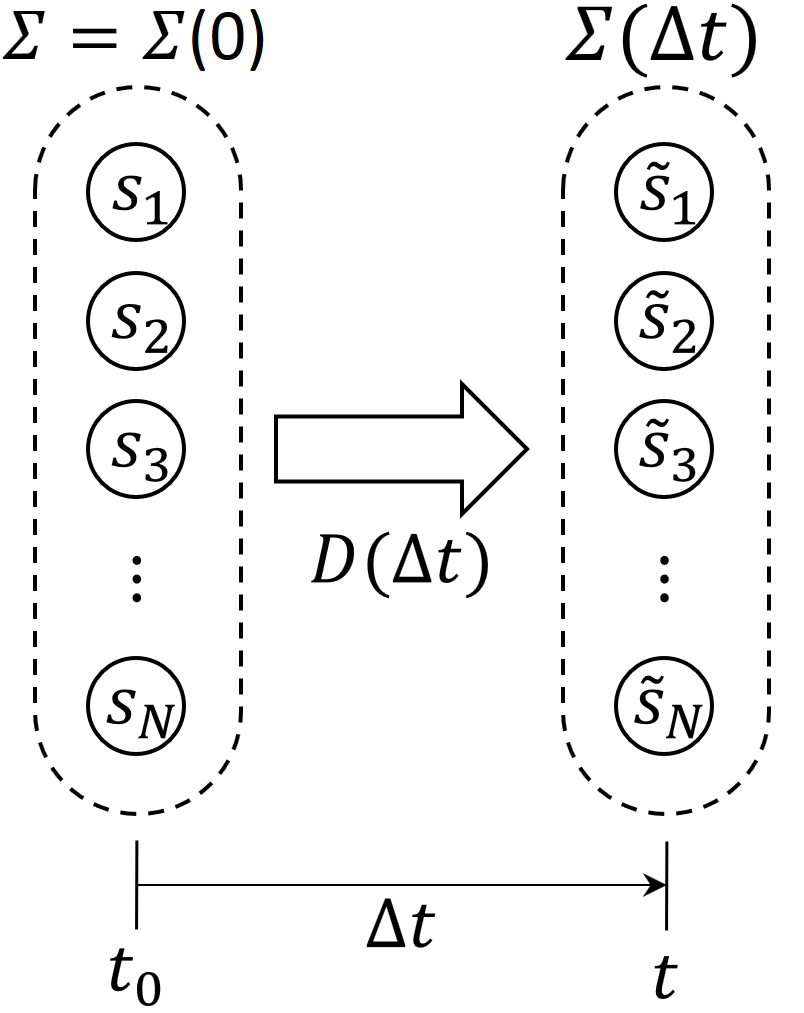}
\caption{For any given amount $\Delt t\in\mathbb{R}$ of elapsed time, 
a bijective dynamics $D$ gives us a one-to-one map $D(\Delt t)$ from 
the state space $\Sigma$ at time $t_0$ to a new state space 
$\Sigma(\Delt t)$ at time $t=t_0+\Delt t$.  Each of the states 
$s_i\in\Sigma$ at time $t_0$ is mapped by $D$ to a unique new state 
$\tilde s_i$ at time $t$.}
\label{fig:bijdyn}
\end{figure}

The assumption that $D(\Delt t)$ is a single-valued function for 
positive values of $\Delt t$ implies that the dynamics is 
\emph{deterministic} (meaning, the state at any future time is 
determined by the present state at $t_0$), and the assumption that 
$D(\Delt t)$ is single-valued for negative values of $\Delt t$ implies
that it is ``reverse-deterministic'' or \emph{reversible}, meaning, 
the state at any past time can be determined from the present state.  

Also, although here we did not specifically require $D$ to exhibit 
general time-translation symmetry (that is, to always have exactly the 
same form, regardless of the initial time $t_0$), the fact that, for 
any $\Delt t$, the function $D(\Delt t)$ is one-to-one implies that it 
has a corresponding inverse function $D^{-1}(\Delt t)$, and therefore, 
$D$ also induces a bijective map
$$
D'(\Delt t_1,\Delt t_2) 
	= D(\Delt t_2)\circ D^{-1}(\Delt t_1)\; : \;
		\Sigma(\Delt t_1)\rightarrow \Sigma(\Delt t_2)
$$
between the states at any pair of times $t_0+\Delt t_1$ and $t_0+\Delt 
t_2$, where $\Delt t_1,\Delt t_2 \in \mathbb{R}$.  Thus, even if the 
dynamics $D$ was reexpressed relative to any different reference time 
point $t'_0\neq t_0$, then, even if it didn't retain exactly the same 
form under that transformation, it would, at least, remain bijective.

Finally, we need a concept of the completeness of state-space 
descriptions of physical systems.  For our purposes, we will take 
``physical system'' itself to be a primitive, undefined concept.

\begin{defn} \label{def:sscomp} \textbf{\emph{Completeness of state 
spaces.}}  We say that a state space $\Sigma$ representing a set of 
distinguishable states of a particular physical system $\Pi$ at
some point in time $t_0$ is \emph{complete} if and only if there
is no larger state space $\Sigma'$, \emph{i.e.} with 
$|\Sigma'|>|\Sigma|$, that also describes $\Pi$.
\end{defn}

The point of the concept of the completeness of the state space is 
just to say that the state representation of the physical state is
fully detailed, \emph{i.e.}, that its states are not actually 
composite entities that could be factored into aggregates of more
fundamental distinguishable states.  It is perhaps an open 
philosophical question about physics whether we can ever really know
with certainty that a given state-space description of a physical
system is really complete; however, we generally assume that there is
always some state-space description that is at least complete with 
respect to all of the ways to probe a system that have been discovered 
at a given point in time.

Now, given the above definitions, we can state the following assertion, 
which we consider to be a solidly-established fact about all of the 
viable modern theories of fundamental physics, and which can be 
considered to be the basic postulate upon which the proof of 
{\lanprinc} rests.

\begin{assertion} \label{ass:bijdyn} \textbf{\emph{Bijectivity of 
physical dynamics.}}  All viable modern theories of fundamental 
physics (\emph{i.e.}, all those that are empirically well-founded, 
logically consistent, and parsimonious) exhibit the property that for 
any closed (isolated) physical system $\Pi$, if we characterize it by 
some complete state space $\Sigma$ at some reference time $t_0$, the 
time-evolution of that system (over all future and past times) is 
described by some bijective dynamics $D$ on $\Sigma$.
\end{assertion}

Assertion~\ref{ass:bijdyn} definitely holds in the case of all viable 
quantum theories, which share the property that the system's dynamics 
is implicitly determined by some time-independent, rank-$N$ 
Hamiltonian operator $H$ (an energy-valued Hermitian linear 
operator), from which we can derive a unitary time-evolution 
operator
$$
U(\Delt t) = \mathrm{e}^{-\mathrm{i}H\Delt t/\hbar},
$$
and the evolved state space $\Sigma(\Delt t)$ at time $t_0+\Delt t$ is 
then 
just
$$
\Sigma(\Delt t) = \{U(\Delt t)\ket{s_i} \; : \; i\in \{1,...,N\}\},
$$
where $\ket{s_i}$ denotes a representation of state $s_i$ as a quantum 
state vector; \emph{e.g.}, expressed in a basis where the states
$s_i$ correspond to basis vectors, this could simply be a rank-$N$ column 
vector, whose $i^{\mathrm{th}}$ element is 1, and other elements 0.  The map 
determined by the dynamics between the state spaces at different times 
is then just
$$
[D(\Delt t)](s_i) = U(\Delt t)\ket{s_i}.
$$

The above formulation covers standard quantum mechanics, and also 
(if we extend it to infinite-dimensional state spaces) all of the 
standard relativistic quantum field theories, which can successfully 
model all of the known fundamental physical forces except for gravity.  
Although, at this time, we do not yet have a complete and well-tested 
theory of physics that succeeds at unifying gravity (\emph{i.e.}, 
general relativity) with quantum mechanics, it is generally assumed by 
physicists that whenever we do find such a theory, it will still 
exhibit the same general properties above that are shared by all of 
the existing quantum theories.

It's important to note that if physical dynamics was not reversible, 
then the Second Law of Thermodynamics would not be true; in detail, if 
the dynamical map $D(\Delt t)$ was many-to-one for any positive 
elapsed times $\Delt t > 0$, then formerly-distinguishable states 
could merge together, and entropy would spontaneously decrease.  So, 
from this perspective, the reversibility of real physical dynamics 
follows from the empirical observation that the Second Law does not 
appear to be violated.  Likewise, the determinism of the dynamics can 
be inferred from empirical observations showing \emph{e.g.} the 
effectiveness of Schr\"odinger's deterministic wave equation at 
modeling the observed dynamics of closed quantum systems.

In any event, if one accepts the bijectivity of dynamical evolution as 
a truism of mathematical physics, then the validity of Landauer's 
Principle follows rigorously from it, as a theorem.  To state and prove 
that theorem formally, we will require just a few more definitions.

%---------1---------2---------3---------4---------5---------6---------7
\subsection{Computational state spaces.}
\label{sec:compstate}

Here, we define what we mean by a computational, as opposed to 
physical, state space.  Such a distinction is necessary because 
computational states are typically considered to be abstract, 
higher-level entities; we do not typically consider that what is 
important in our description of a computer includes the complete, 
fully-detailed physical dynamics of the physical system implementing 
that computer.

\begin{defn} \label{def:compss} \textbf{\emph{Computational subspaces, 
computational states.}} Given a state space $\Sigma$, a 
\emph{computational subspace} $C$ of $\Sigma$ is a partition of the 
set $\Sigma$, \emph{i.e.}, a set of non-overlapping, non-empty subsets 
of $\Sigma$ whose union is $\Sigma$.  (See Figure~\ref{fig:comp-ss} 
for an illustration.)  We say that a physical system $\Pi$ is \emph{in 
computational state} $c_j\in C$ whenever the physical state of the 
system is not reliably distinguishable from some $s_i$ such that 
$s_i\in c_j$.
\end{defn}

\begin{figure}
\centering
\includegraphics[height=2.5in]{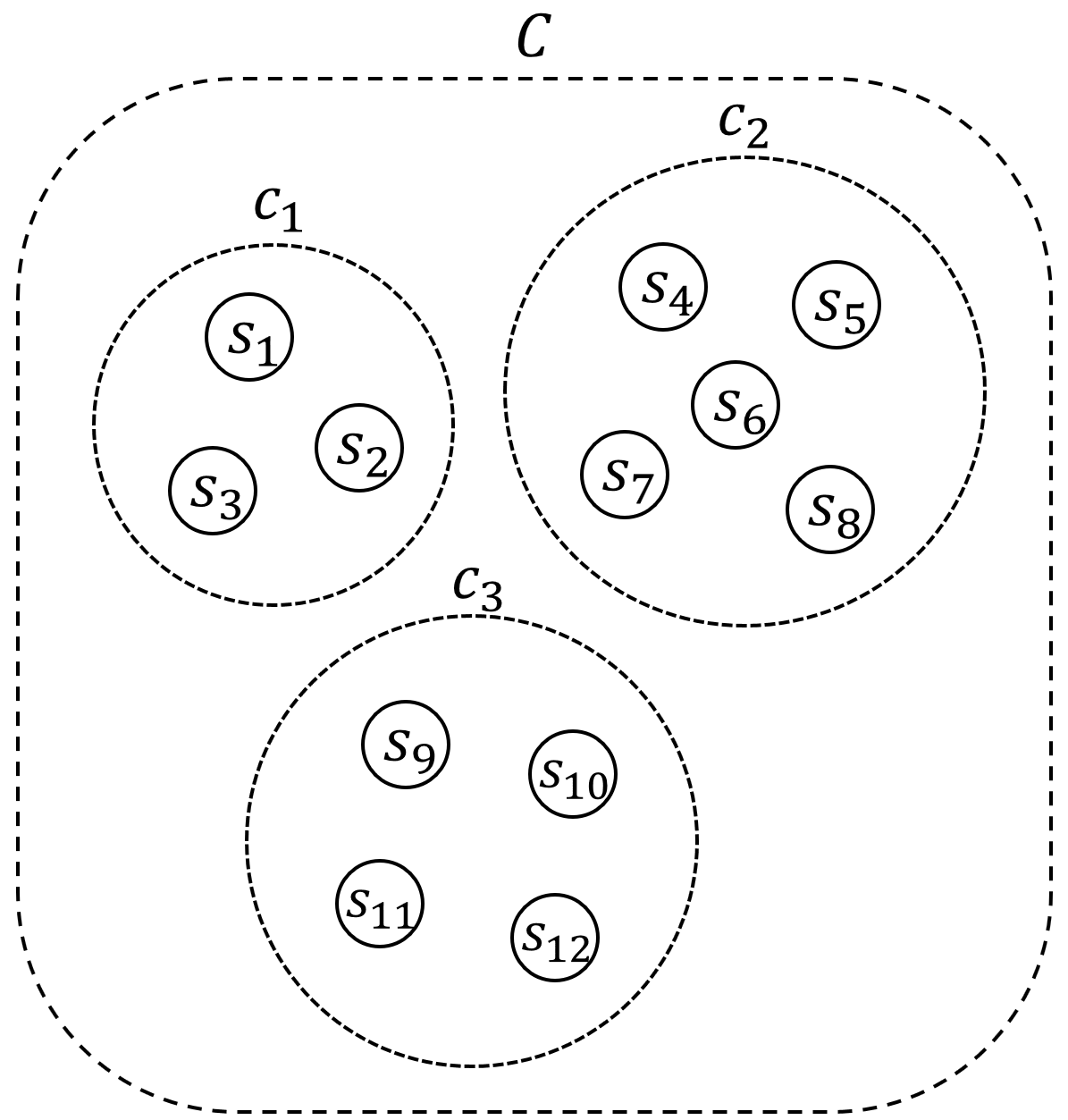}
\caption{An example of a physical state space 
	$\Sigma=\{s_1,s_2,...,s_{12}\}$ with $N=12$ states, partitioned into
	a computational subspace $C=\{c_1, c_2, c_3\}$ with $M=3$ states,
	$c_1 = \{ s_1, s_2, s_3\}$, $c_2 = \{ s_4, s_5, ..., s_8 \}$, 
	and	$c_3 = \{ s_9, s_{10}, s_{11}, s_{12} \}$.
}
\label{fig:comp-ss}
\end{figure}

The idea of a computational subspace is simply that not all features
of a system's physical state are computationally meaningful; for 
example, the detailed microscopic state of a heat sink surrounding a 
computer would not generally be considered to be part of the machine's 
computational state.  Two states that are, in principle, 
distinguishable from each other physically, but that are not 
considered to be distinct in terms of their computational 
interpretation, would be considered to be part of the same 
computational state $c_j$.

Note that according to this definition, a system could be considered 
to be in more than one computational state at the same time, in the 
case where it takes on a physical state that is not reliably 
distinguishable from either of two different physical states $s_1, 
s_2$ that happen to be contained in different elements of the 
partition $C$.  This would be the case, for example, in a quantum 
computer that has been carefully prepared so as to occupy a quantum 
superposition of two distinct computational states (and such 
situations are, in fact, required for execution of quantum algorithms).  
However, for our purposes in the present paper, we will assume that 
real computational systems are normally intentionally designed to be 
highly decoherent systems in which the physical states that are used 
to assemble computational states correspond to naturally-stable 
``pointer'' states, as in \cite{zurek}.  Under this assumption, 			% [12]
superpositions of computational states will, in practice, be extremely 
rare; therefore, we normally assume that a system can be considered to 
only occupy one computational state at a time, with probability 
approaching 1.  A more comprehensive version of the theoretical model 
presented in this paper would relax that restriction.

Next, let us assume that we can also identify an appropriate 
computational subspace $C(\Delt t)$ which is a partition of the 
evolved physical state space $\Sigma(\Delt t)$ at any past or future 
time $t_0 + \Delt t\in\mathbb{R}$.  If we model the computer itself as 
being assembled or disassembled over the timeline, the size of the 
computational subspace $C(\Delt t)$ might change over time, but that 
will not materially affect any aspect of our subsequent discussion.

%---------1---------2---------3---------4---------5---------6---------7
\subsection{Probability distributions, entropy measures.}
\label{sec:probent}

Consider, now, any initial-state probability distribution $p_0$ over 
the complete state space $\Sigma=\Sigma(0)$ at time $t=t_0$, that is, 
a real-valued function
$$
p_0: \Sigma(0)\rightarrow [0,1]
$$
such that the state probabilities $p_0(s_i)$ sum to unity, $\sum_i 
p_0(s_i) = 1$.  This then clearly induces an implied initial 
probability distribution $P_\mathrm{I}$ over the \emph{computational} 
states at time $t_0$ as well:
$$
P_\mathrm{I}(c_j) = \sum_{k=0}^{|c_j|} p_0(s_{j,k}),
$$
where $s_{j,k}$ denotes the $k$th physical state in computational 
state $c_j\in C$.

For any probability distributions $p$ and $P$ over physical and 
computational states, we can then define corresponding entropy 
measures:

\begin{defn} \label{def:physent} \textbf{\emph{Physical entropy.}}
Given any probability distribution $p$ over a physical state space 
$\Sigma$, the \emph{physical entropy} $S(p)$ is defined by:
$$
S(p) = \sum_{i=0}^{N=|\Sigma|} p(s_i)\log\frac{1}{p(s_i)},
$$
\end{defn}
where the logarithm there can be considered to be an indefinite 
logarithm, dimensioned in generic logarithmic units, or, if we wish to 
express the result in particular logarithmic units such as bits (log 
base 2 units) or ``nats'' (log base e units) we can substitute a 
definite logarithm for the indefinite one as follows:
$$
\log x = (1\,\mathrm{bit})\log_2 x 
       = (1\,\mathrm{nat})\log_\mathrm{e} x 
	   = k_\mathrm{B} \ln x 
$$
where note that Boltzmann's constant $k_\mathrm{B}$ here can be 
considered to simply represent 1 nat or the natural logarithmic unit 
(log e), interpreted as being a physical unit of 
entropy.\footnote{Boltzmann's constant is more familiarly defined in 
terms of energy/temperature units, such as $k_\mathrm{B} = 1.38 
\times 10^{-23}\,\mathrm{J}/\mathrm{K}$, but another way of 
understanding the meaning of this formula is simply to say that 1 
degree Kelvin of temperature is equivalent to $1.38 \times 
10^{-23}\,\mathrm{J}$ per natural-log unit of entropy.  In any case, 
the identity $k_\mathrm{B} = 1\,\mathrm{nat}$ can be considered to be 
a simple factual statement that reflects the fundamentally 
information-theoretic nature of physical entropy.}

The above definition of physical entropy comports with the standard 
quantum-mechanical concept of the von Neumann entropy of a mixed 
quantum state $\rho$; if one diagonalizes the density-matrix 
description of any mixed quantum state, the probabilities $p(s_i)$ of 
the basis states $\ket{s_i}$ lie along the diagonal, and the standard 
definition of Von Neumann entropy reduces to the definition above.

The bijectivity of physical dynamics then implies the following 
theorem:

\begin{thm} \label{thm:entcons} \textbf{\emph{Conservation of entropy.}}
The physical entropy of any closed system, as determined for any 
initial state distribution $p_0$, is exactly conserved over time, 
under any viable physical dynamics.  

In other words, under any of the viable theories of fundamental 
physics mentioned in Assertion~\ref{ass:bijdyn}, if the physical 
entropy of an initial-state distribution $p_0(s_i)$ at time $t_0$ is 
$S(0)$, and we evolve that system over an elapsed time $\Delt t \in 
\mathbb{R}$ according to its bijective dynamics $D(\Delt t)$, 
the physical entropy $S(\Delt t)$ of its final-state probability 
distribution $p_{\Delt t}$ that applies at time $t_0 + \Delt t$ 
will be the exact same value, $S(\Delt t)=S(0)$.
\end{thm}

\begin{proof}
Since the dynamical map from the state at time $t_0$ to the state at 
time $t_0+\Delt t$ is a one-to-one function, the probability 
distributions $p_0$ and $p_{\Delt t}$ comprise the exact same bag 
(multiset) of real numbers (albeit reassigned to new states); 
therefore, the entropy values of these distributions are identical. 
\qed 
\end{proof}

It's a standard theorem of quantum theory that the von Neumann entropy 
$S(\rho)$ of any mixed quantum state $\rho$ is conserved under any 
unitary time-evolution operator $U(\Delt t)$; Theorem~\ref{thm:entcons} 
can be considered to be a generalization of that standard theorem that 
does not depend on the full mathematical structure of quantum theory, 
but only on the bijectivity of its dynamics.

It's important to note that the validity of Theorem~\ref{thm:entcons} 
(entropy is conserved) depends on the assumption that, from our 
theoretical perspective, we know (and in principle, can track) the 
dynamical evolution of the state exactly---if we did not (for example, 
if we didn't know the laws of physics exactly, or if we discarded some 
information about the state as the system evolved), then we would
accumulate increased uncertainty about the final state, compared to 
the initial state, and so entropy would be seen to increase in 
practice, which is what we in fact observe.  But, it remains true that 
\emph{if} we knew the precise dynamics, and could track the evolution 
exactly, the entropy of a given state distribution would not be seen 
to increase at all as that distribution is evolved by the dynamics.

Next, we can define what is sometimes called the ``information entropy'' 
of the computational state:

\begin{defn} \label{def:infent} \textbf{\emph{Information entropy.}}
Given any probability distribution $P$ over a computational state 
space $C$, the \emph{information entropy} or \emph{computational 
entropy} $H(P)$ is defined by:
$$
H(P) = \sum_{j=0}^{|C|} P(c_j)\log\frac{1}{P(c_j)},
$$
\end{defn}
which, like $S(p)$, is dimensioned in arbitrary logarithmic units.

Note that further, this is really \emph{the exact same definition} as 
for physical entropy $S(p)$, except that here, we are just applying it 
to a different probability distribution, namely, the one induced over 
the computational states.  The fact that different logarithmic units 
are most conventionally chosen for computational versus physical 
entropy (bits versus nats, respectively) is just an historical 
accident, an arbitrary difference in units, and is totally 
inconsequential.  Entropy is entropy!  

One reason, however, why we might sometimes want to use the name 
``information entropy'' for this concept is that the information 
contained in the computational state might, in principle, be known 
information (and thus, not ``true'' entropy at all!) if the history 
of how that information was computed from other known information is 
known.  However, from the perspective of an individual computational 
device that does not have access to that kind of nonlocal knowledge, 
an uncertain statistical description of the state remains appropriate, 
and thus, the information in the computational state is effectively 
still entropy, for purposes of taking a thermodynamics of computation
perspective towards the analysis of local device operations.

Finally, we can define the ``non-computational entropy'' as comprising 
the remainder of the total physical entropy, other than the 
computational part:

\begin{defn} \label{def:ncent} \textbf{\emph{Non-computational entropy.}}
Given a situation where the total physical entropy is $S$ and the
computational entropy is $H$, the \emph{non-computa\-tional entropy} 
is defined by:
$$
S_{\mathrm{nc}} = S - H.
$$
\end{defn}
It is clear that $S_{\mathrm{nc}}\geq0$ always, since the summing of 
probabilities that occurs in aggregating physical states to form 
computational states can only reduce the entropy; the computational 
entropy can, thus, never be greater than the total physical entropy.

To understand why non-computational entropy is physically meaningful,
it is helpful to consider the following theorem:

\begin{thm} \label{thm:physnc} \textbf{\emph{Physical role of 
non-computational entropy.}}  Non-com\-pu\-ta\-tion\-al entropy is the 
physical entropy conditioned on the computational information.  In 
other words,
$$
S_{\mathrm{nc}} = S(s|c),
$$
where $S(s|c)$ denotes the conditional entropy of random variable $s$
(the physical state) when the value of random variable $c$ (the 
computational state) is known.

Conditional entropy is itself defined, in general, by:
$$
S(x|y) = \sum_j p(y_j) S(p(x\;|\;y=y_j)),
$$
or in other words, as the weighted average, given the probability 
distribution $p(y_j)$ over possible values $y_j$ of random variable 
$y$, of the entropies of the probability distribution $p(x\;|\;y=y_j)$ 
over $x$, conditioned on the given value of $y$.
\begin{proof}
Since the physical state $s$ determines the computational state $c$, 
specifying the physical state $s$ is the same as jointly specifying 
$s$ and $c$, and so the statement of the theorem boils down to a 
special case of the chain rule of conditional entropy,
$$
S(s|c) = S(s,c) - S(c),
$$
since the conditional entropy $S(s|c)$ is stated by the theorem to
equal $S_{\mathrm{nc}} = S - H = S(s) - S(c) = S(s,c) - S(c)$. \qed
\end{proof}
\end{thm}

The practical import of this theorem is simply to clarify that the 
non-com\-pu\-ta\-tion\-al entropy is not just some arbitrary, meaningless 
quantity, rather, it is the expected value of the physical entropy 
in contexts where the computational entropy is really known 
information, which will commonly be the case, whenever the 
computational state is computed deterministically from other known 
information.  Therefore, it has a physical significance.  Increased 
non-computational entropy means increased physical entropy from the 
user's perspective (assuming that the user cannot keep track of the 
detailed physical state, but only, at most, the computational state).

%---------1---------2---------3---------4---------5---------6---------7
\subsection{Statements of {\lanprinc}.}
\label{sec:landprinc}

The above definitions then suffice to allow us to formulate and prove 
{\lanprinc}, in both its most general quantitative form, as well as in 
another form more frequently seen in the literature.

\begin{thm} \label{thm:genland} \textbf{\emph{{\lanprinc} (general 
formulation).}}  If the entropy of the computational state of a system 
at initial time $t_0$ is $H_\mathrm{I} = H(P_\mathrm{I})$, and we allow 
that system to evolve, according to its physical dynamics, to some 
other ``final'' time $t_0 + \Delt t$, at which its computational 
entropy becomes $H_\mathrm{F} = H(P_\mathrm{F})$, where $P_\mathrm{F} 
= P(\Delta t)$ is the induced probability distribution over the 
computational state set $C(\Delt t)$ at time $t_0+\Delt t$, then the 
non-computational entropy is increased by
$$
\mathrm{\Delta}S_{\mathrm{nc}} = H_\mathrm{I} - H_\mathrm{F}.
$$
\begin{proof}
Total physical entropy is conserved by Theorem~\ref{thm:entcons}.  The 
computational part of the total entropy decreases by $H_\mathrm{I} - 
H_\mathrm{F}$, by hypothesis.  Therefore, the noncomputational part 
(the remainder) must increase by that amount. \qed
\end{proof}
\end{thm}

That formulation captures the conceptual core of {\lanprinc}, and as 
we can see, its proof is really extremely simple, essentially trivial.  
It amounts to the simple observation that, since total physical 
entropy is dynamically conserved, then any decrease in computational 
entropy must cause a corresponding increase in non-computational 
entropy.  

Furthermore, conventional digital devices are typically designed to 
locally reduce computational entropy, \emph{e.g.}, by erasing 
``unknown'' old bits obliviously (that is, without utilizing 
independent knowledge of their previous value) or (in other words) by 
destructively overwriting them with newly-computed values.  As a 
result, typical device operations necessarily eject entropy into the 
non-computational form, and so, over time, non-computational entropy 
typically builds up in the system (manifesting as heating), but, we 
generally assume that it cannot build up indefinitely in the system 
(since eventually the physical mechanism would break down from 
overheating), but must instead eventually be moved out into some 
external thermal environment at some temperature $T$, which involves 
the dissipation of energy $\Delt E_\mathrm{diss} = 
T\mathrm{\Delta}S_\mathrm{nc}$ to the form of heat in that environment, 
by the very definition of thermodynamic temperature,
$$
\frac{1}{T} = \frac{\partial S}{\partial Q},
$$  
where $S$ and $Q$ are respectively the entropy and heat energy content 
of a heat bath at thermodynamic equilibrium.  From 
Theorem~\ref{thm:genland} together with these facts, along with the 
logarithmic identity $1\,\mathrm{bit} = (1\,\mathrm{nat})/\log_2 
\mathrm{e} = k_\mathrm{B}\ln 2$, follows the more commonly-seen 
statement of {\lanprinc}:

\begin{cor} \label{cor:comland}
\emph{{\lanprinc} (common form).}  For each bit's worth of information 
that is lost within a computer (\emph{e.g.}, by obliviously erasing or 
destructively overwriting it), an amount of energy
$$
E_\mathrm{diss} = k_\mathrm{B}T\ln 2
$$
must eventually be dissipated to the form of heat added to some 
environment at temperature $T$.
\begin{proof}
See preceding discussion.
\end{proof}
\end{cor}

Next, we will see that, if we examine a little more precisely what are 
the conditions for ``entropy ejection'' by device operations, we find 
that such an event, in the case of logically deterministic 
computational operations, must occur whenever we operate a mechanism 
that is designed to transform either of (at least) two initial local
computational states that each have nonzero probability into the same 
final computational state---since doing so will result in a local
reduction in the computational entropy.  We will develop these ideas 
further in the next section.  But, as we will see, the qualifier 
``have nonzero probability'' in the preceding statement turns out to 
be essential, and this, in fact, is what makes the difference between 
the traditional theoretical model of reversible computing, and the 
more generalized framework developed here.

%---------1---------2---------3---------4---------5---------6---------7
\section{Reformulating Reversible Computing Theory}
\label{sec:rcthy}

The conceptual development of Generalized Reversible Computing (GRC) 
theory rests on a process of very carefully and thoroughly analyzing 
the implications of {\lanprinc} (in its general formulation above) 
for computation.

Carrying out such an analysis allows us, first, to formally verify 
what we call the \emph{Fundamental Theorem of Traditional Reversible 
Computing Theory} (Theorem~\ref{thm:tradrc} below), which states that 
deterministic computational operations that are always 
non-entropy-ejecting, independently of their statistical operating 
context, must be (unconditionally) logically reversible.  

However, we can then go further in our analysis, and also prove 
a new \emph{Fundamental Theorem of Generalized Reversible Computing 
Theory} (Theorem~\ref{thm:genrc} below), which demonstrates that 
additionally, a computation that applies any arbitrary deterministic 
operation (which is, in general, only what we call 
\emph{conditionally} logically reversible) \emph{within any specific 
operating context that satisfies any of the preconditions for the 
reversibility of that operation} is also non-entropy-ejecting, 
according to {\lanprinc}.  This then establishes that it is actually 
the more general concept of conditional reversibility, rather than the 
more restrictive traditional concept of unconditional reversibility, 
that is the most general concept of logical reversibility that is 
consistent with the requirement of avoiding ejection of entropy from 
the computational state under {\lanprinc}.  

The key insight that allows us to advance from the traditional theory 
to the more general one is simply the realization, discussed in 
\cite{deben}, that a computation \emph{per se} consists of not 				% [13]
\emph{just} a choice of an abstract computational operation to be 
carried out, but also a specific statistical operating context in 
which that operation is to be applied.  Without considering the 
specific statistical context, one cannot calculate the initial and 
final computational entropies with any degree of accuracy, and 
therefore, one cannot correctly infer whether any entropy is in fact 
ejected from the computational state in the course of the computation.  

The traditional theory of reversible computing, ever since Landauer's 
original paper, has typically neglected to emphasize this fact, 
leading to what is arguably an overemphasis on the more restrictive 
model of unconditionally-reversible computing that is invoked 
throughout the majority of the reversible computing literature.  
By exploring the more general, conditional model, we can design 
simpler hardware mechanisms for reversible computing than can be 
modeled within the traditional framework, as we will show in section 4.
Therefore, arguably the general model deserves more intensive 
attention and study than it has, to date, received.

In the following two subsections, we first formally re-develop the 
traditional theoretical foundations of reversible computing, and then 
show how to extend those foundations to support the generalized model.

%---------1---------2---------3---------4---------5---------6---------7
\subsection{Traditional Theory of Unconditional Reversibility}
\label{sec:tradthy}

Here, we redevelop the foundations of the traditional theory of 
unconditionally logically-reversible operations, using a language that 
we can subsequently build upon to develop the generalized theory.  We 
begin by defining some basic concepts of computational devices and 
operations, explain what it means for computational operations to be 
deterministic and reversible, define what we mean by a statistical 
operating context, and what it means to say that a deterministic 
computational operation is unconditionally reversible, which is the 
traditional notion of logical reversibility.  We can then show why 
unconditional logical reversibility is indeed necessary if we wish to 
always avoid ejecting entropy from the computational to the 
non-computational state \emph{independently} of the statistical 
properties that apply within the context of a specific computation.  
That much is then sufficient, as a basic foundation for traditional 
reversible computing theory.

Then, in subsection~\ref{sec:genthy}, we will go further, and show how 
to develop the more general, context-dependent theory.

%---------1---------2---------3---------4---------5---------6---------7
\subsubsection{Computational devices and operations.}
\label{sec:devops}

First, let us clarify what we mean by a computational device.

\begin{defn} \label{def:dev} \textbf{\emph{Devices.}}
For our purposes, a computational \emph{device} $D$ will simply be any 
physical artifact that is capable of carrying out one or more 
different \emph{computational operations} (to be defined).
\end{defn}
We generally assume that the scale of any given device is 
circumscribed, in the sense that it is associated with some 
physical and computational state information that is localizable; 
for example, this could include the states of some I/O terminals 
incident on the device, as well as internal states of the device.  
Although devices are not, strictly speaking, closed systems (since 
they will generally exchange energy and information/entropy with their 
environment), we generally assume, for our purposes, that the 
information-theoretic thermodynamics of individual device operations 
can be analyzed more or less in isolation from other external systems; 
that is the significance of saying that devices are ``localizable.''

In line with the definitions of the previous section, devices can be
considered to have physical and computational state spaces $\Sigma,C$ 
associated with their (assumed-localizable) state.  In general, the 
identity of these sets could change over time for a given device, but 
that aspect of the situation will not be particularly important to our
present analysis.

Next, we define computational operations, that is, operations that
are intended to possibly transform the computational state.  These 
can generally include computational operations that are deterministic,
nondeterministic, reversible, or irreversible; we'll clarify the 
meanings of these terms momentarily.

\begin{defn} \label{def:op} \textbf{\emph{Computational operations.}}
Given a device $D$ with an associated initial local computational 
state space $C_\mathrm{I}=\{c_{\mathrm{I}1}, ..., c_{\mathrm{I}m}\}$ 
at some point in time $t_0$, a \emph{computational operation $O$ on D} 
that is applicable at $t_0$ is specified by giving a probabilistic 
transition rule, \emph{i.e.}, a stochastic mapping from the initial 
computational state at $t_0$ to the final computational state at some 
later time $t_0 + \Delt t$ (with $\Delt t>0$) by which the operation 
will have been completed.  Let the computational state space at this 
later time be $C_\mathrm{F}=\{c_{\mathrm{F}1}, ..., c_{\mathrm{F}n}\}$.
Then, the operation $O$ is a map from $C_\mathrm{I}$ to probability 
distributions over $C_\mathrm{F}$; which is characterizable, in terms 
of random variables $c_\mathrm{I},c_\mathrm{F}$ for the initial and 
final computational states, by a conditional probabilistic transition 
rule
$$
r_i(j) = \mathrm{Pr}(c_\mathrm{F} = c_{\mathrm{F}j}\;|\; 
					  c_\mathrm{I} = c_{\mathrm{I}i}),
$$
where $i\in\{1,...,m\}$ and $j\in\{1,...,n\}$.  That is, $r_i(j)$ 
denotes the conditional probability that the final computational state 
is $c_{\mathrm{F}j}$, given that the initial computational state is 
$c_{\mathrm{I}i}$.
\end{defn}
Note that, if we specify a computational operation in this way, in 
terms of transition probabilities between computational states, this 
does not, by itself, say anything about the initial probability 
distribution over physical or computational states, except that the 
initial probability distribution over physical states must be one for
which implementing the desired transition rule $r_i(j)$ is possible.
Although we will not detail this argument here, as long as there are 
many physical states per computational state, and the detailed 
physical state is allowed to equilibrate with the thermal environment 
in between computational operations to the extent that the 
distribution over physical states within each computational state 
approaches an equilibrium distribution, this will generally be a 
requirement that is possible to satisfy. 

%---------1---------2---------3---------4---------5---------6---------7
\subsubsection{Deterministic and reversible operations.}
\label{sec:detrev}

Now, for later reference, let us define the concepts of determinism
and (unconditional logical) reversibility of computational operations.

\begin{defn} \label{def:detop} \textbf{\emph{Deterministic and 
nondeterministic operations.}} A computational operation $O$ will be 
called \emph{deterministic} if and only if all of the probability 
distributions $r_i$ are single-valued.  In other words, for each 
possible value of the initial-state index $i\in\{1,...,m\}$, there is 
exactly one corresponding value of the final-state index $j$ such that 
$r_i(j)>0$ (and thus, for this value of $j$, it must be the case that
$r_i(j)=1$), while $r_i(k)=0$ for all other $k\neq j$.  If an 
operation $O$ is not deterministic, we call it \emph{nondeterministic}.
As a notational convenience, for a deterministic operation $O$, we can
write $O(c_\mathrm{Ii})$ to denote the $c_\mathrm{F}j$ such that 
$r_i(j)=1$, that is, treating $O$ as a simple transition function 
rather than a stochastic one.
\end{defn}

Note that this is a different sense of the word ``nondeterministic'' 
than is commonly used in computational complexity theory, when 
referring to, for example, nondeterministic Turing machines, which 
conceptually evaluate all of their possible future computational 
trajectories in parallel.  Here, when we use the word 
``nondeterministic,'' we mean it simply in the physicist's sense, to 
refer to ``randomizing'' or ``stochastic'' operations.

\begin{defn} \label{def:revop} \textbf{\emph{Reversible and 
irreversible operations.}}  A computational operation $O$ will be 
called (unconditionally logically) \emph{reversible} if and only if 
all of the probability distributions $r_i$ have non-overlapping 
nonzero ranges.  In other words, for each possible value of the 
final-state index $j\in\{1,...,n\}$, there is at most one 
corresponding value of the initial-state index $i$ such that 
$r_i(j)>0$, while $r_k(j)=0$ for all other $k\neq i$.  If an 
operation $O$ is not reversible, we call it \emph{irreversible}.
\end{defn}

Essentially, the above definition is just a statement that the 
transition relation 
$$
R=\{(i\in\{1,...,m\},\:j\in\{1,...,n\})\;\;|\;\;r_i(j)>0\}
$$
specifying which initial states have nonzero probability of 
transitioning to which final states is an injective relation.  
(Note, however, that $R$ may not be a functional relation, if the 
operation is nondeterministic.)

Now, up to this point, the notion of ``reversible'' that we have
invoked here is essentially the same concept of (what we call 
\emph{unconditional}) logical reversibility that has been used ever 
since Landauer.  And this is, indeed, the appropriate concept for 
considering the reversibility of computational operations in the 
abstract, independently of any particular statistical context in 
which they may be operating.  If we wish for a deterministic 
computational operation to avoid ejecting entropy into the 
non-computational state no matter what the initial-state distribution 
is, then it must be an unconditionally reversible operation.  This was 
already observed by Landauer.  Let us now set up some more definitions
so that we can prove this formally.  

%---------1---------2---------3---------4---------5---------6---------7
\subsubsection{Operating contexts and entropy-ejecting operations.}
\label{sec:opcon}

First, we define what we mean by a statistical operating context:

\begin{defn} \label{def:opcon} \textbf{\emph{Operating contexts.}}
For a computational operation $O$ with an initial computational state 
space $C_\mathrm{I}$, a (statistical) \emph{operating context} for 
that operation is any probability distribution $P_\mathrm{I}$ over the 
initial computational states; for any $i\in\{1,...,m\}$, the value of 
$P_\mathrm{I}(c_{\mathrm{I}i})$ gives the probability that the initial
computational state is $c_{\mathrm{I}i}$.
\end{defn}

And now, we define the concept of an operation that may eject entropy 
from the computational state:

\begin{defn} \label{def:eeop} \textbf{\emph{Entropy-ejecting 
operations.}}  A computational operation $O$ is called (potentially) 
\emph{entropy-ejecting} if and only if there is some operating context 
$P_\mathrm{I}$ such that, when the operation $O$ is applied within 
that context, the increase $\Delt S_\mathrm{nc}$ in the 
non-computational entropy required by {\lanprinc} is greater than zero.  
If an operation $O$ is not potentially entropy-ejecting, we call it 
\emph{non-entropy-ejecting}.
\end{defn}

%---------1---------2---------3---------4---------5---------6---------7
\subsubsection{Fundamental theorem of traditional reversible computing.}
\label{sec:tradthm}

Now, we can formally prove Landauer's original result stating that 
only operations that are logically reversible (in his sense) can 
always avoid ejecting entropy from the computational state 
(independently of the operating context).

\begin{thm} \label{thm:tradrc} \textbf{\emph{(Fundamental Theorem of 
Traditional Reversible Computing) Non-entropy-ejecting deterministic 
operations must be re\-ver\-si\-ble.}}  If a deterministic 
computational operation $O$ is non-entropy-ejecting, then it is 
reversible in the sense defined above (its transition relation is 
injective).
\begin{proof}
Suppose that $O$ is not reversible.  Then by definition it is 
possible to find a final-state index $j$ such that there at least 
two initial state indices (which we identify as $i=1,2$ without loss 
of generality) such that both $r_1(j)>0$ and $r_2(j)>0$.  Assuming 
that the operation is also deterministic, we must have that 
$r_1(j)=r_2(j)=1$.  Thus, all of the probability mass assigned to 
$c_{\mathrm{I}1}$ and $c_{\mathrm{I}2}$ by the initial-state 
probability distribution $P_\mathrm{I}$ will get mapped onto the 
same final state $c_{\mathrm{F}j}$.  Therefore, if we let the 
operating context $P_\mathrm{I}$ assign probability $1/2$ to each of 
states $c_{\mathrm{I}1}$ and $c_{\mathrm{I}2}$, then the initial 
computational entropy $H_\mathrm{I} = 1\,\mathrm{bit}$, and the final 
computational entropy $H_\mathrm{F} = 0$, and thus the entropy ejected 
is $\Delt S_\mathrm{nc} = 1\,\mathrm{bit}$.  Thus, the operation $O$ 
is potentially entropy-ejecting.  Therefore, by contraposition, if a 
deterministic computational operation is non-entropy-ejecting, then it 
must be reversible. \qed
\end{proof}
\end{thm}

Note that in this proof, we chose an operating context that assigned 
probability 1/2 to the two initial states that are being merged, to 
construct an example demonstrating that the hypothesized irreversible 
operation is entropy-ejecting.  Actually, however, any nonzero 
probabilities for these two states would have sufficed; the 
contribution to the initial computational entropy from those two 
states would then be greater than zero, which was all that was needed 
to show that the operation was entropy-ejecting.  But, if only one of 
the two initial states being merged had been assigned nonzero 
probability, the proof would not have gone through.  This is the key 
realization that sets us up to develop the more general framework of 
Generalized Reversible Computing.  To formalize that framework, we 
need some more definitions.  

%---------1---------2---------3---------4---------5---------6---------7
\subsection{General Theory of Conditional Reversibility}
\label{sec:genthy}

To develop the generalized theory, we define a notion of a 
\emph{computation}, which fixes a specific statistical operating
context for a computational operation, and then we examine the 
detailed requirements for a given computation to be 
non-entropy-ejecting.  This leads to the concept of \emph{conditional
reversibility}, which is the most general concept of logical 
reversibility, and thus provides the appropriate foundation for the 
generalized theory.

%---------1---------2---------3---------4---------5---------6---------7
\subsubsection{Computations and entropy-ejecting computations.}
\label{sec:comp}

First, we define our specific technical notion of a \emph{computation}, 
by which we mean a given computational operation \emph{performed within 
a specific operating context}.

\begin{defn} \label{def:comp} \textbf{\emph{Computations.}}
For our purposes, a \emph{computation} $\mathcal{C}=(O, P_\mathrm{I})$ 
performed by a device $D$ is defined by specifying \emph{both} a 
computational operation $O$ to be performed by that device, \emph{and} 
a specific operating context $P_\mathrm{I}$ under which the operation 
$O$ is to be performed.
\end{defn}

Now, for specific computations, which carry an associated statistical 
operating context, as opposed to computational operations considered
in a more generic, context-free way, we obtain a new, 
context-dependent notion of what it means for such a computation to be 
entropy-ejecting under {\lanprinc}, and an associated new 
context-dependent notion of logical reversibility (namely, conditional 
reversibility) that corresponds to it.

\begin{defn} \label{def:eecomp} \textbf{\emph{Entropy-ejecting 
computations.}}  A computation $\mathcal{C}=(O,P_\mathrm{I})$ is 
called (specifically) entropy-ejecting if and only if, when the 
operation $O$ is applied within the specific operating context 
$P_\mathrm{I}$, the increase $\Delt S_\mathrm{nc}$ in the 
non-computational entropy required by {\lanprinc} is greater than zero.  
If $\mathcal{C}$ is not specifically entropy-ejecting, we call it 
\emph{non-entropy-ejecting}.
\end{defn}

%---------1---------2---------3---------4---------5---------6---------7
\subsubsection{Conditionally-reversible operations.}
\label{sec:condrev}

In order to characterize the set of deterministic computations that can
be non-entropy-ejecting, we need to consider a certain class of 
computational operations, which we call the conditionally-reversible 
operations:

\begin{defn} \label{def:crop} \textbf{\emph{Conditionally-reversible 
computational operations.}} A deterministic computational operation $O$ 
is called {\emph{conditionally reversible}} if and only if there is a 
non-empty subset $A\subseteq C_\mathrm{I}$ of initial computational 
states (the \emph{assumed set}) that $O$'s transition rule maps onto 
an equal-sized set $B\subseteq C_\mathrm{F}$ of final states.  (Each 
$c_{\mathrm{I}i}\in A$ maps, one to one, to a unique 
$c_{\mathrm{F}j}\in B$ where $r_i(j)=1$.)  We say that $B$ is the 
\emph{image of $A$ under $O$}.  We also say that $O$ is 
\emph{conditionally reversible under the precondition that the initial 
state is in $A$}, or that $A$ \emph{is a precondition under which $O$ 
is reversible}.
\end{defn}

It turns out that \emph{all} deterministic computational operations 
are, in fact, conditionally reversible, under some 
sufficiently-restrictive preconditions.

\begin{thm} \label{thm:detcr} \textbf{\emph{Conditional reversibility 
of all deterministic operations.}}  All deterministic computational 
operations are conditionally reversible.
\begin{proof}
Given a deterministic computational operation $O$ with initial 
computational state space $C_\mathrm{I}$, (which must be non-empty, 
since it is a partition of a non-empty set of physical states), 
consider any initial computational state $c_{\mathrm{I}i}\in 
C_\mathrm{I}$, and let $A=\{c_{\mathrm{I}i}\}$, the singleton set of 
state $c_{\mathrm{I}i}$.  Since the operation is deterministic, 
$r_i(j)=1$ for a single final-state index $j$, so then let 
$B=\{c_{\mathrm{F}j}\}$.  Both $A$ and $B$ are the same size (both 
singletons), and thus $O$ is conditionally reversible under the 
precondition that the initial state is in $A$. \qed
\end{proof}
\end{thm}

Note that, although that proof of Theorem~\ref{thm:detcr} invoked 
singleton sets for simplicity, any deterministic operation that has a 
larger number $K>1$ of different final computational states that are 
reachable (that is, that have initial states that transition to them) 
is reversible under at least one precondition set $A$ that is of size 
$K$.  To build such an $A$, it suffices, for each reachable final 
state, to include in $A$ any single one of the initial states that 
transitions to it.  Therefore, all deterministic operations that have 
more than one reachable final computational state are conditionally 
reversible in this nontrivial sense, and not just in the more trivial 
sense that we invoked in the proof of Theorem~\ref{thm:detcr}.

%---------1---------2---------3---------4---------5---------6---------7
\subsubsection{Conditioned reversible operations.}
\label{sec:conditioned}

Whenever we wish to fix a \emph{specific} assumed precondition $A$ for 
the reversibility of a conditionally-reversible operation $O$, we use
the following concept:

\begin{defn} \label{def:condop} \textbf{\emph{Conditioned reversible 
computational operations.}}  Let $O$ be any conditionally-reversible 
computational operation, and let $A$ be any of the preconditions under 
which $O$ is reversible.  Then the \emph{conditioned reversible 
operation} $O_A = (O, A)$ denotes the concept of performing operation 
$O$ in the context of a requirement that precondition $A$ is satisfied.  
Furthermore, for any given computation $\mathcal{C}=(O,P_\mathrm{I})$, 
we can say that it \emph{satisfies the assumed condition for 
reversibility of $O_A$ with probability $P$} if
$$
P=\sum_{c_j\in A} P_\mathrm{I}(c_j).
$$ 
\end{defn}

%---------1---------2---------3---------4---------5---------6---------7
\subsubsection{Fundamental theorem of generalized reversible computing.}
\label{sec:genthm}

The central result of generalized reversible computing theory 
(Theorem~\ref{thm:genrc}, below) is then that \emph{any} deterministic 
computation $\mathcal{C}=(O,P_\mathrm{I})$ will be non-entropy-ejecting 
(as a computation), and therefore, will avoid any requirement under 
{\lanprinc} to dissipate any amount of energy greater than zero to its 
thermal environment, as long as its operating context $P_\mathrm{I}$ 
assigns probability 1 to any precondition $A$ under which its 
computational operation $O$ is reversible.  

Moreover (as we see in Theorem~\ref{thm:eevan} later), even if the 
probability of satisfying some such precondition only \emph{approaches} 
1, this is sufficient for the entropy ejected (and energy required to
be dissipated) to approach 0.

\begin{thm} \label{thm:genrc} \textbf{\emph{(Fundamental Theorem of 
Generalized Reversible Computing) Any deterministic computation is 
non-entropy-ejecting if and only if at least one of its preconditions 
for reversibility is satisfied.}}  \emph{I.e.}, let $\mathcal 
C=(O,P_\mathrm{I})$ be any deterministic computation (\emph{i.e.}, any 
computation whose operation $O$ is deterministic).  Then, part (a):  
If there is some precondition $A$ under which $O$ is reversible, such 
that $A$ is satisfied with certainty in the operating context 
$P_\mathrm{I}$, or in other words, such that
$$
\sum_{c_{\mathrm{I}i}\in A} P_\mathrm{I}(c_{\mathrm{I}i}) = 1,
$$
then $\mathcal C$ is a non-entropy-ejecting computation.  And, part 
(b):  Alternatively, if no such precondition $A$ is satisfied with 
certainty, then $\mathcal C$ is entropy-ejecting.
\begin{proof}
Part (a):  All of the probability mass in the initial-state 
distribution $P_\mathrm{I}$ falls into initial computational states 
within the set $A$, which (since $A$ is a precondition for the 
reversibility of $O$) are mapped, one-to-one, onto an equal number of 
final computational states by the transition rule $r_i(j)$.  
Therefore, the final-state distribution $P_\mathrm{F}$ consists of the 
same bag of probability values as the initial-state distribution, and 
therefore its computational entropy is the same, 
$H_\mathrm{F}=H_\mathrm{I}$.  Thus, $\Delt S_\mathrm{nc} = 0$, and so 
$\mathcal C$ is non-entropy-ejecting. Part (b):  In contrast, if none 
of $O$'s preconditions for reversibility is satisfied, then this 
implies that the set $X=\{c_{\mathrm{I}i} \; | \; 
P_\mathrm{I}(c_{\mathrm{I}i}) > 0\}$ of initial states with nonzero 
probability is mapped by $O$ to a smaller set of final states (since 
otherwise $X$ would itself be a precondition for reversibility that is 
satisfied), and therefore, by the pigeonhole principle, at least two 
initial states $c_{\mathrm{I}1},c_{\mathrm{I}2}$ in $X$ must be mapped 
(and their probability mass carried) to the same final state 
$c_{\mathrm{F}j}$, and since both $c_{\mathrm{I}1}$ and 
$c_{\mathrm{I}2}$ have nonzero probability, and the function $h(p) = 
p\log  p^{-1}$ that we sum over to compute $H(\cdot)$ is subadditive, 
$c_{\mathrm{F}j}$'s contribution to the total final computational 
entropy $H_\mathrm{F}$ is less than the total that $c_{\mathrm{I}1}$ 
and $c_{\mathrm{I}2}$ contributed to $H_\mathrm{I}$ before they were 
merged; and therefore, the final computational entropy $H_\mathrm{F}$ 
must be less than the initial computational entropy $H_\mathrm{I}$, 
and thus, the entropy ejected $S_\mathrm{nc}>0$. \qed
\end{proof}
\end{thm}

The upshot of Theorem~\ref{thm:genrc} is that, in order for it to be 
possible for some device to carry out a given computation in an 
asymptotically thermodynamically reversible way (with entropy 
generated approaching zero, and energy dissipated to the environment 
approaching zero), it is {\emph{not}} necessary for the computational 
operation being performed to be one that is unconditionally logically 
reversible; rather, it is {\emph{only}} necessary (if the operation 
is deterministic) that there be some precondition for the 
reversibility of the operation that is satisfied with probability 1, 
or approaching 1, in the specific operating context in which that 
device will be performing that operation.  For this to work in general 
(with dissipation approaching 0), the device must be designed with 
implicit knowledge of not only what conditionally-reversible operation 
it should perform, but also which specific one of the preconditions 
for that operation's reversibility should be assumed to be satisfied.

The following theorem shows why it's sufficient, for asymptotic 
thermodynamic reversibility, for the probability of satisfying a 
given precondition $A$ to only \emph{approach} 1.

\begin{thm} \label{thm:eevan} \textbf{\emph{Entropy ejection vanishes 
as precondition certainty approaches unity.}} Let $O$ be any 
deterministic operation, and let $A$ be any precondition under which 
$O$ is reversible, and let $P_{\mathrm{I}1}, P_{\mathrm{I}2}, ...$ be 
any sequence of operation contexts for $O$ within which the total 
probability mass assigned to $A$ approaches 1.  Then, in the 
corresponding sequence of computations, the entropy ejected $\Delt 
S_\mathrm{nc}$ also approaches 0.
\begin{proof}
Consider, without loss of generality, any pair 
$c_{\mathrm{I}1},c_{\mathrm{I}2}$ of initial states that are both 
mapped by the operation $O$ to the same final state $c_{\mathrm{F}j}$, 
where $c_{\mathrm{I}1}\in A$, while $c_{\mathrm{I}2}\notin A$; and 
letting $p=P_{\mathrm{I}\ell}(c_{\mathrm{I}1})$ and 
$q=P_{\mathrm{I}\ell}(c_{\mathrm{I}2})$; and let $r=p/q$ be the ratio 
between the probabilities of these two initial states.  Let $\Delt s$ 
be the contribution of this state merger to the total entropy $\Delt 
S_\mathrm{nc}$ ejected to the non-computational state.  An analytical 
derivation % Should I write out the steps in this extended version?
based on the definitions of $H_\mathrm{I}$ and 
$H_\mathrm{F}$ then shows that the following expression for the 
convergence of $\Delt s$ is accurate to first order in $r$, as $r$ 
increases:
$$
\Delt s \rightarrow \frac{p}{r}(1 + \ln r)\,\mathrm{nat}.
$$
And, it's easy to see that the value of this expression itself 
approaches 0, almost in proportion to $q=p/r$ as $r\rightarrow\infty$ 
and $q\rightarrow 0$.  We can then consider applying this observation 
to each initial state that is not in $A$ that merges with some state 
in $A$.  Moreover, for any other states $c_{\mathrm{I}i}$ not in $A$ 
that may merge with each other but not with any state in $A$, their 
individual contributions $\Delt s$ to the total entropy ejected are 
upper-bounded by their contributions $h_i=q\log q^{-1}$ to the initial 
computational entropy, where $q$ again is their probability, and this 
$h_i\rightarrow 0$ as $q\rightarrow 0$.  It is then clear that, as 
$\ell\rightarrow\infty$, and the total probability of all of the 
states satisfying the precondition approaches 1 from below, the total 
probability of all the states violating the precondition falls to 0, 
and so does an upper bound on each of their individual probabilities 
$q$, and thus on each of their contributions to the entropy ejected, 
and thus (recalling that the state set is finite) their total 
contribution to the ejected entropy falls to 0, and the theorem holds. 
\qed
\end{proof}
\end{thm}

A numerical example illustrating how the $\Delt S_\mathrm{nc}$ 
calculation comes out in a specific case where the probability of
violating the precondition for reversibility is small can be found 
in~\cite{deben}, which also discussed the more general issue that the 		% [13*]
initial-state probabilities must be taken into account in order to 
properly apply {\lanprinc}.

%---------1---------2---------3---------4---------5---------6---------7
\subsubsection{Reversals of conditioned reversible operations.}
\label{sec:reversal}

As we saw in Theorem~\ref{thm:detcr}, any deterministic computational 
operation $O$ is con\-di\-tionally-reversible with respect to any 
given one $A$ of its suitable preconditions for reversibility.  For 
any computation $\mathcal{C}=(O,P_\mathrm{I})$ that satisfies the 
conditions for reversibility of the conditioned reversible operation 
$O_A$ with certainty, we can undo the effect of that computation 
exactly by applying any conditioned reversible operation that is what 
we call a \emph{reversal} of $O_A$.  At an intuitive level, the 
reversal of a conditioned reversible operation is simply an operation 
that maps the image of the assumed set back onto the assumed set 
itself in a way that exactly inverts the original forward map.  We can 
define this more formally as follows:

\begin{defn} \label{def:reversal} \textbf{\emph{Reversals of 
conditioned reversible operations.}}  Let $O_A$ be any conditioned 
reversible operation with assumed set $A$, and let $B\subseteq C_F$ be 
the image of $A$ under $O$.  Then a {\emph{reversal of $O_A$}} is any 
conditioned reversible operation $O'_{B'}$ where $B'\supseteq B$, and 
where the image of $B'$ under $O'$ is similarly a superset of $A$, and 
where
$$
\forall c_i\in A:\; O'(O(c_i)) = c_i.
$$
That is to say, for any initial computational state $c_i$ in the 
assumed set $A$, whatever is the final state $O(c_i)$ that $O$ maps it 
to, the operation $O'$ maps that state back to the original state $c_i$.  
In other words, $O'_{B'}$ exactly undoes whatever transformation $O_A$ 
applied to the original assumed set $A$.  
\end{defn}

%---------1---------2---------3---------4---------5---------6---------7
\subsubsection{Incorporating nondeterminism.}
\label{sec:nondet}

The above definitions and theorems can also be extended to work with 
nondeterministic computations.  In fact, adding nondeterminism to an 
operation only makes it easier to avoid ejecting entropy to 
the noncomputational state, since nondeterminism tends to increase the 
computational entropy, and thus tends to reduce the noncomputational 
entropy.  As a result, a nondeterministic operation can be 
non-entropy-ejecting (or even entropy-absorbing, \emph{i.e.}, with 
$\Delt S_\mathrm{nc} < 0$) even in computations where none of its 
preconditions for reversibility are satisfied, so long as the 
reduction in computational entropy caused by its irreversibility is 
compensated for by an equal or greater increase in computational 
entropy caused by its nondeterminism.  An example of such an operation
was given in {\cite{ismvl}}.  However, we will not take the time, in 		% [8*]
the present paper, to flesh out detailed analyses of such cases.

%---------1---------2---------3---------4---------5---------6---------7
\section{Examples of Conditioned Reversible Operations}
\label{sec:examples}

In this section, we define and illustrate a number of examples of
conditionally-reversible operations (including a specification of 
their assumed preconditions) that comprise natural primitives to
use for the composition of more complex reversible algorithms.

First, we introduce some textual and graphical notations for
describing conditioned reversible operations.

%---------1---------2---------3---------4---------5---------6---------7
\subsection{Notations}
\label{sec:notations}

First, let the computational state space be factorizable into 
independent \emph{state variables} $x, y, z, ...$, which are in 
general $n$-ary discrete variables.  Common special cases will 
be binary variables ($n=2$).  For simplicity, we will assume for
purposes of this section that the sets of state variables into 
which the initial and final computational state spaces are 
factorized are identical, although more generally this may not 
be the case.  A convenient method for describing 
conditionally-reversible operations together with their assumed
preconditions is to use language specifying initial conditions 
on the state variables, and how those variables are transformed.

\begin{notation}
Given a computational state space $C$ that is factorizable into state 
variables $x, y, z, ...$, and given a precondition $A$ on the initial 
state defined by $$A=\{c_i\in C\;|\;P(x, y, ...)\},$$ where 
$P(x, y, ...)$ is some propositional (\emph{i.e.}, Boolean-valued) 
function of the state variables $x, y, ...$, we can denote a 
conditionally-reversible operation $O_A$ on $C$ that is reversible
under precondition $A$ using notation like:
$$
\mathtt{OpName}(x, y, ...\;|\;P(x, y, ...))
$$
which represents a conditionally-reversible operation named 
\texttt{OpName} that operates on and potentially transforms the state 
variables $x, y, ...$, and that is reversible under an assumed 
precondition $A$ consisting of the set of initial states that satisfy 
the given proposition $P(x,y,...)$.  
\end{notation}

In the above notation, the proposition $P(x,y,...)$ for the assumed 
precondition may sometimes be left implicit and omitted; however, when 
this is done, readers should keep in mind that any particular device 
intended to be capable of carrying out a given conditionally reversible 
operation in an asymptotically physically reversible way will 
nevertheless necessarily have, built into its design, some particular 
choice of an assumed precondition with respect to which its physical 
operation will in fact be asymptotically physically reversible.  

\begin{notation}
A simple, generic graphical notation for a deterministic, 
con\-di\-tion\-al\-ly-reversible operation named \texttt{OpName}, 
operating on a state space that is decomposable into three state 
variables $x,y,z$, and possibly including an assumed precondition for 
reversibility $P(x,y,z)$, is the ordinary space-time diagram 
representation shown in Fig.~\ref{fig:xyz-op}.
\begin{figure}
\centering
\includegraphics[height=1in]{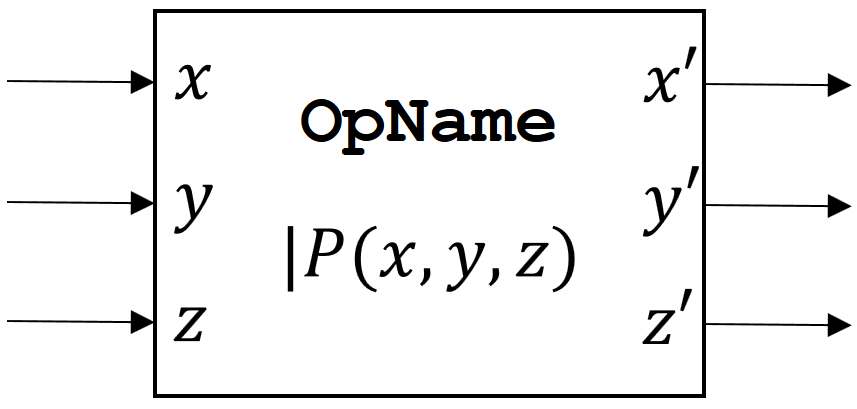}
\caption{Generic graphical notation for a deterministic, 
conditioned reversible operation $\mathtt{OpName}(x, y, 
z\,|\,P(x, y, z))$ on three state variables $x,y,z$, with an assumed 
precondition specified by the propositional function $P(x,y,z)$.}
\label{fig:xyz-op}
\end{figure}

In this representation, as in standard reversible logic networks,
time is imagined as flowing from left to right, and the horizontal 
lines represent state variables.  The primed versions $x',y',z'$ 
going outwards represent the values of the state variables in the
final computational state $c_\mathrm{F}$ after the operation is 
completed.
\end{notation}

%---------1---------2---------3---------4---------5---------6---------7
\subsection{Reversible set and reset.}
\label{sec:setreset}

As Landauer observed, operations such as ``set to one'' and ``reset 
to zero'' on binary state spaces are logically irreversible, under his 
definition; indeed, they constitute classic examples of \emph{bit 
erasure} operations for which (assuming equiprobable inputs) an amount 
$k_\mathrm{B}\ln 2$ of entropy is ejected from the computational state.  
However, as per Theorem~\ref{thm:detcr}, these operations are in fact
conditionally reversible, under suitably-restricted preconditions.  A
suitable precondition, in this case, is one in which one of the two
initial states is required to have probability 0, in which case, the 
other state must have probability 1.  In other words, the initial 
state is known with certainty in any operating context satisfying 
such a precondition.  A known state can be transformed to any specific 
new state reversibly.  If the new state is different from the old one, 
such an operation is non-vacuous.  Thus, we have the following 
conditioned reversible operations that are useful:

%---------1---------2---------3---------4---------5---------6---------7
\subsection{Reversible set-to-one (\texttt{rSET}).}
\label{sec:setto1}

\begin{defn}
The deterministic operation \texttt{rSET} on a binary variable $x$, 
which (to be useful) is implicitly associated with an assumed 
precondition for reversibility of $x=0$, is an operation that is
defined to transform the initial state into the final state 
$x'=1$; in other words, it performs the operation $x:=1$.  Standard 
and simplified graphical notations for this operation are 
illustrated in Figure~\ref{fig:rset}.
\begin{figure}
\centering
\includegraphics[height=0.5in]{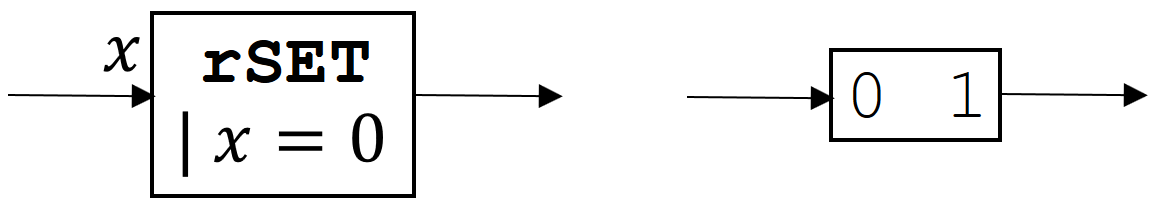}
\caption{(Left) Standard graphical notation for the conditioned
reversible operation $\mathtt{rSET}(x\;|\;x=0)$; (Right) Simplified 
symbol.}
\label{fig:rset}
\end{figure}
\end{defn}

By Theorem~\ref{thm:genrc}, the conditioned reversible operation 
$\mathtt{rSET}(x|x=0)$ is specifically non-entropy-ejecting in 
operating contexts where the designated precondition for 
reversibility is satisfied.  It can be implemented in a way that 
is asymptotically physically reversible (as the probability that 
its precondition is satisfied approaches 1) using any mechanism 
that is designed to adiabatically transform the state $x=0$ to 
the state $x=1$.

%---------1---------2---------3---------4---------5---------6---------7
\subsection{Reversible reset-to-zero (\texttt{rCLR}).}
\label{sec:resetto0}

\begin{defn}
The deterministic operation \texttt{rCLR} on a binary variable 
$x$, which (to be useful in a reversible mode) is implicitly 
associated with an assumed precondition for reversibility of 
$x=1$, is an operation that is defined to transform the initial 
state into the final state $x'=0$; in other words, it performs 
the operation $x:=0$.  Standard and simplified 
graphical notations for this operation are illustrated in 
Figure~\ref{fig:rclr}.
\begin{figure}
\centering
\includegraphics[height=0.5in]{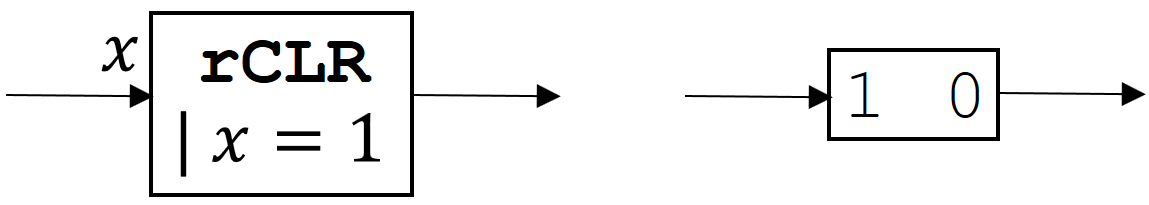}
\caption{(Left) Standard graphical notation for the 
conditionally-reversible operation $\mathtt{rCLR}(x\;|\;x=1)$; (Right) 
Simplified symbol.}
\label{fig:rclr}
\end{figure}
\end{defn}

Like with $\mathtt{rSET}$, the $\mathtt{rCLR}(x\;|\;x=1)$ operation is 
non-entropy-ejecting when its designated precondition for 
reversibility is satisfied.  It can be implemented in a way that is 
asymptotically physically reversible (as the probability that its 
precondition is satisfied approaches 1) using any mechanism that is 
designed to adiabatically transform the state $x=1$ to the state 
$x=0$.

It should be noted that, whenever the precondition for the 
reversibility of an \texttt{rSET} or \texttt{rCLR} operation
is satisfied, the outcome of the operation will be identical
to the outcome that would be obtained from a traditional,
unconditionally-reversible \texttt{cNOT} (controlled-\texttt{NOT}) 
operation, $x:=\bar{x}$).  However, \texttt{rSET} and \texttt{rCLR}
may be simpler to implement than \texttt{cNOT}; for example, one
can implement them both by simply connecting a circuit node to 
a voltage reference that then adiabatically transitions between 
the \texttt{0} and \texttt{1} logic levels in the appropriate 
direction.  Whereas, to perform an in-place \texttt{cNOT} operation
requires more steps than this in an adiabatic switching circuit.
\emph{E.g.}, we could first reversibly copy $x$'s value to a temporary 
node, then use this to control the charging/discharging of $x$ to
its new level, then decompute the copy based on the new value of 
$x$.  This illustrates how using the traditional, unconditionally 
reversible paradigm increases hardware complexity.

%---------1---------2---------3---------4---------5---------6---------7
\subsection{Reversible set-to-$i$ (\texttt{rSET}$i$).}
\label{sec:settoi}

Similarly to \texttt{rSET}/\texttt{rCLR}, if we are given a state 
variable $x$ having any higher arity $n>2$, there are reversible 
``set to $i$'' operations for larger result values $i=2,3,...$; 
however, when $n>2$, there is a choice among multiple possible 
non-vacuous preconditions.  The general form of the conditioned 
reversible set-to-$i$ operations is thus $$\mathtt{rSET}i_j = 
\mathtt{rSET}i(x\;|\;x=j),$$ where $j\neq i$ is the assumed initial 
value of the variable $x$, and $i$ is the final value to which it is
being set.  The graphical notation of ordinary \texttt{rSET} can be 
generalized appropriately.

%---------1---------2---------3---------4---------5---------6---------7
\subsection{Reversible copy and uncopy.}
\label{sec:copyuncopy}

A very commonly-used computational operation is to copy one state 
variable to another.  As with any other deterministic operation, such 
an operation will be conditionally reversible, under suitable 
preconditions.  An appropriate precondition for the reversibility of 
this \texttt{rCOPY} operation is any in which the initial value of the 
target variable is known, so that it can be reversibly transformed to 
the new value.  A standard reversal of a suitably-conditioned 
\texttt{rCOPY} operation, which we can call \texttt{rUnCopy}, is 
simply a conditioned reversible operation that transforms the final 
states resulting from \texttt{rCOPY} back to the corresponding initial 
states.

\begin{defn}
Let $x,y$ be any two discrete state variables both with the same arity 
(number $n$ of possible values, which without loss of generality we may 
label $0,1,...$), and let $v\in\{0,1,...,n-1\}$ be any fixed initial 
value.  Then \emph{reversible copy of $x$ onto $y=v$} or 
$$\mathtt{rCOPY}_v = \mathtt{rCOPY}(x,y\;|\;y=v)$$ is a conditioned 
reversible operation $O$ with assumed precondition $y=v$ that maps any 
initial state where $x=i$ onto the final state $x=i, y=i$.  In the 
language of ordinary pseudocode, the operation performed is simply 
$y:=x$.
\end{defn}

\begin{defn}
Given any conditioned reversible copy operation $\mathtt{rCOPY}_v$, 
there is a conditioned reversible operation which we hereby call 
\emph{reversible uncopy of $x$ from $y$ back to $v$} or 
$$
\mathtt{rUnCopy}_v = \mathtt{rUnCOPY}_v(x,y\;|\;y=x)
$$
which, assuming (as its precondition for reversibility) that initially 
$x=y$, carries out the operation $y:=v$, restoring the destination 
variable $y$ to the same initial value $v$ that was assumed by the 
\texttt{rCOPY} operation.
\end{defn}

Figure~\ref{fig:rcopy} below shows graphical notations for 
$\mathtt{rCOPY}_v$ and $\mathtt{rUnCOPY}_v$.  It is easy to see that 
corresponding $\mathtt{rCOPY}$ and $\mathtt{rUnCOPY}$ operations are 
reversals of each other, as was intended.

\begin{figure}
\centering
\includegraphics[height=1.5in]{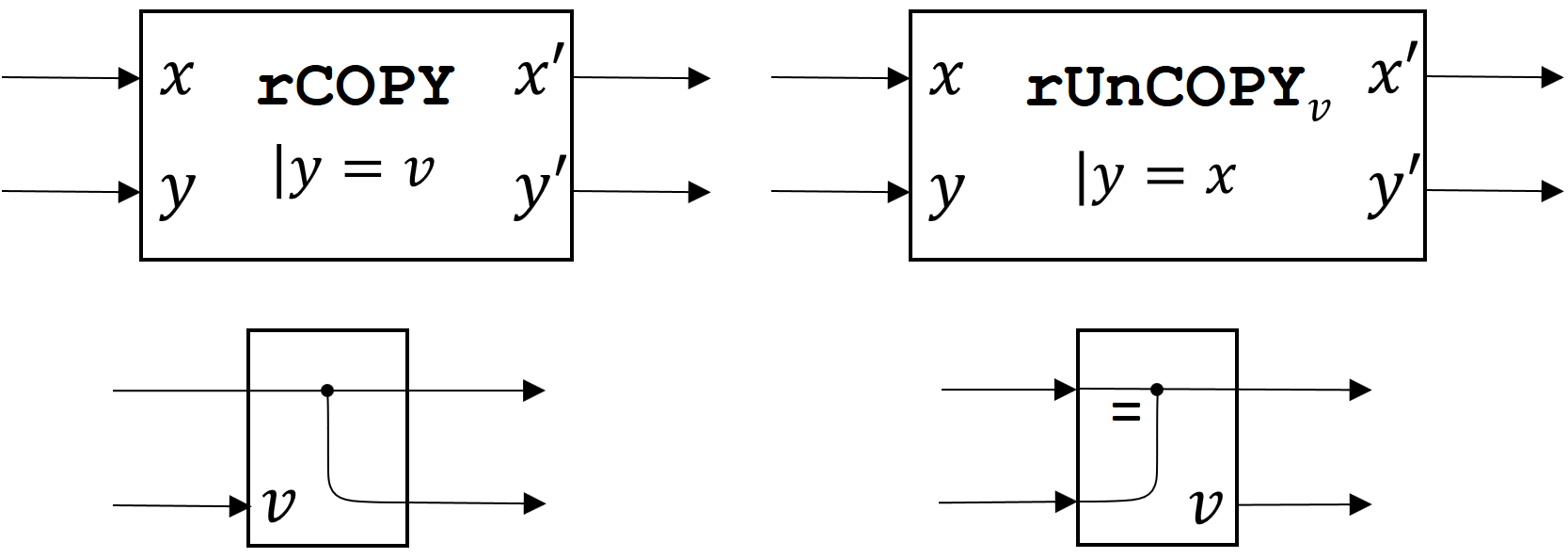}
\caption{(Left) Reversible copy of $x$ onto $y=v$, top: standard 
notation, bottom: simplified symbol; (Right) Reversible uncopy of $x$ 
from $y$ back to $v$, top: standard notation, bottom: simplified 
symbol.}
\label{fig:rcopy}
\end{figure}

\begin{thm} 
$\mathtt{rUnCOPY}_v$ is a reversal of $\mathtt{rCOPY}_v$, and 
vice-versa.
\begin{proof}
Clear by inspection. \qed
\end{proof}
\end{thm}

%---------1---------2---------3---------4---------5---------6---------7
\subsection{Reversible general functions.}
\label{sec:revfunc}

It is easy to generalize $\mathtt{rCOPY}$ to more complex functions.
In general, for any function $F(x,y,...)$ of any number of variables, 
we can define a conditioned reversible operation 
$\mathtt{r}F(x,y,z\;|\;z=v)$ which computes that function, and writes 
the result to an output variable $z$ by transforming $z$ from its 
initial value to $F(x,y,...)$, which is reversible under the 
precondition that the initial value of $z$ is some known value $v$.
Its reversal $\mathtt{rUn}F_v(x,y,z\;|\;z=F(x,y))$ decomputes the 
result in the output variable $z$, restoring it back to the value $v$.  
See Fig.~\ref{fig:rf} below.

\begin{figure}
\centering
\includegraphics[height=2.0in]{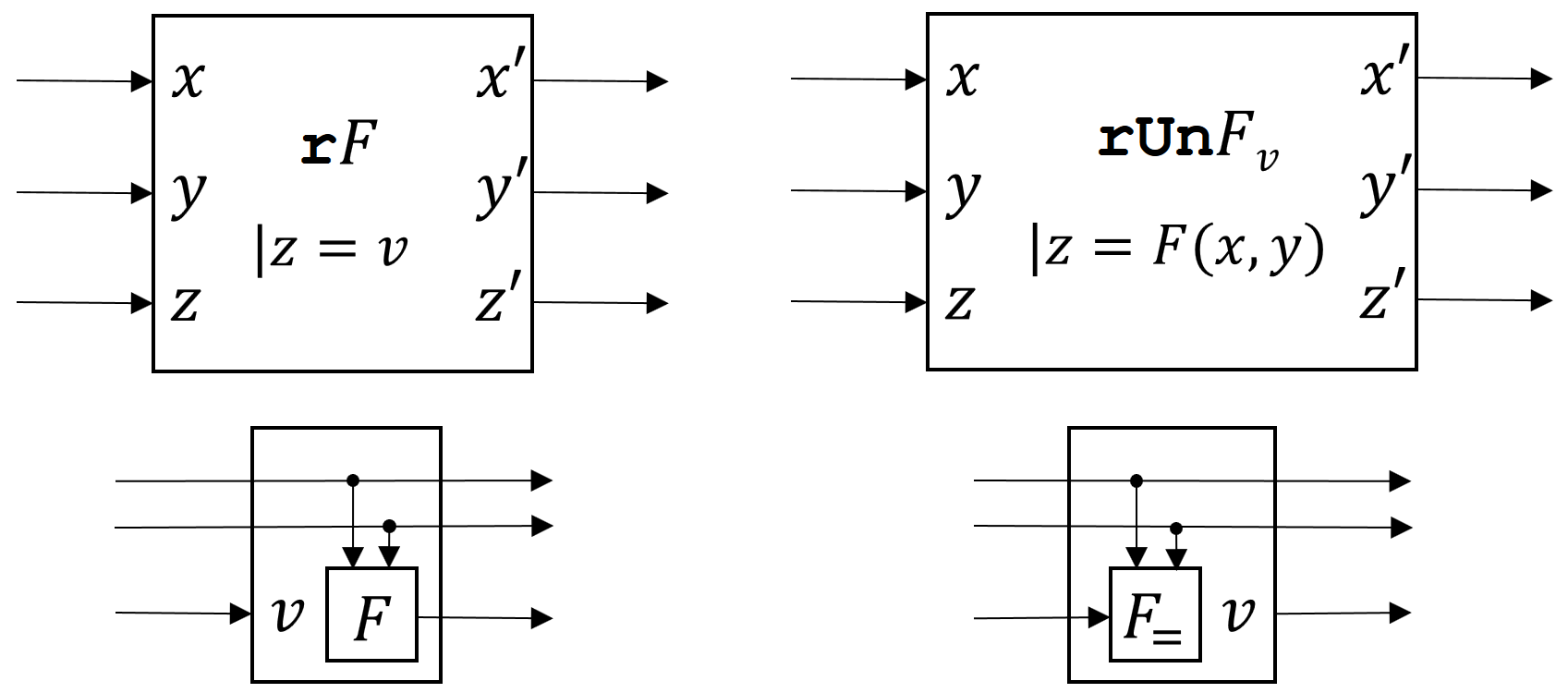}
\caption{Given any function $F(x,y)=z$ of $n$ (here, $n=2$) state 
variables (top), we can easily convert it to a pair of conditioned 
reversible operations $\mathtt{r}F(x,y,z\;|\;z=v)$ and 
$\mathtt{rUn}F_v(x,y,z\;|\;z=F(x,y))$ that are mutual reversals of 
each other that compute and decompute the value of $F$ by reversibly 
transforming the output variable $z$ from and to any predetermined 
value $v$.  The middle row shows standard notation, and the bottom row 
a simplified symbol.}
\label{fig:rf}
\end{figure}

%---------1---------2---------3---------4---------5---------6---------7
\subsection{Reversible Boolean functions.}
\label{sec:revbool}

It's important to note here that $F$ above may indeed be any function,
including standard Boolean logic functions operating on binary 
variables, such as \texttt{AND}, \texttt{OR}, \emph{etc.}  Therefore, 
the above scheme leads us to consider conditioned reversible 
operations such as $\mathtt{rAND}_\mathtt{0}$, 
$\mathtt{rAND}_\mathtt{1}$, $\mathtt{rOR}_\mathtt{0}$, 
$\mathtt{rOR}_\mathtt{1}$; and their reversals 
$\mathtt{rUnAND}_\mathtt{0}$, $\mathtt{rUnAND}_\mathtt{1}$,
$\mathtt{rUnOR}_\mathtt{0}$, $\mathtt{rUnOR}_\mathtt{1}$;
which reversibly do and undo standard \texttt{AND} and \texttt{OR} 
logic operations with respect to output nodes that are expected to be 
a constant logic \texttt{0} or \texttt{1} initially before the 
operation is done (and also finally, after doing the reverse 
operations).

Clearly, one can compose arbitrary functions out of such primitives 
using standard logic network constructions, and later decompute the 
results using the reverse (mirror-image) circuits (after 
\texttt{rCOPY}ing the desired results), following the general 
approach pioneered by Bennett~\cite{benn}.									% [14]

One may wonder, however, what is the advantage of using operations 
such as $\mathtt{rAND}$ and $\mathtt{rUnAND}$ compared to the 
traditional unconditionally reversible operation 
$\mathtt{ccNOT}(x,y,z)$ (controlled-controlled-\texttt{NOT}, a.k.a. 
the Toffoli gate operation \cite{toff}, $z:=z\oplus xy$).  Indeed, 			% [15]
any device that implements $\mathtt{ccNOT}(x,y,z)$ in a 
physically-reversible manner could be used in place of a device
that implements the conditioned reversible operations 
$\mathtt{rAND}(x,y,z\:|\:z=0)$ and 
$\mathtt{rUnAND}_\mathtt{0}(x,y,z\:|\:z=xy)$, or one that 
implements $\mathtt{rNAND}(x,y,z\;|\;z=1)$ and 
$\mathtt{rUnNAND}_\mathtt{1}(x,y,z\;|\;z=\overline{xy})$, in cases 
where the preconditions of those operations would be satisfied.

But, the converse is not true.  In other words, there are 
asymptotically physically reversible implementations of 
$\mathtt{rAND}_\mathtt{0}$ and $\mathtt{rUnAND}_\mathtt{0}$ 
that do not also implement full Toffoli gate operations.  Therefore, 
if what one really needs to do, in one's algorithm, is simply to do 
and undo Boolean \texttt{AND} operations reversibly, then to insist 
on doing this using Toffoli operations rather than conditioned 
reversible operations such as \texttt{rAND} and \texttt{rUnAND} is 
overkill, and amounts to tying one's hands with regards to the 
implementation possibilities, leading to hardware designs that can 
be expected to be more complex than necessary.  Indeed, there are 
very simple adiabatic circuit implementations of devices capable of 
performing \texttt{rAND}/\texttt{rUnAND} and 
\texttt{rOR}/\texttt{rUnOR} operations (based on \emph{e.g.} 
series/parallel combinations of CMOS transmission gates, such as in 
Fig.~\ref{fig:rCopy-circSeq} below), whereas, adiabatic 
implementations of \texttt{ccNOT} itself are typically much less 
simple.  This illustrates our overall point that the Generalized 
Reversible Computing framework generally allows for simpler designs 
for reversible computational hardware than does the traditional 
reversible computing model based on unconditionally reversible 
operations.

%---------1---------2---------3---------4---------5---------6---------7
\section{Modeling Reversible Hardware}
\label{sec:revhw}

A final motivation for the study of Generalized Reversible Computing 
derives from the following observation.

\begin{assertion} \label{ass:adiacr} \textbf{\emph{General 
correspondence between truly, fully adiabatic circuits and conditioned 
reversible operations.}}  Part (a): Whenever a switching circuit is 
operated deterministically in a truly, fully adiabatic way 
(\emph{i.e.}, that asymptotically approaches thermodynamic 
reversibility), transitioning among some discrete set of logic levels, 
the computation being performed by that circuit corresponds to a 
conditioned reversible operation $O_A$ whose assumed precondition $A$ 
is (asymptotically) satisfied.  Part (b): Likewise, any conditioned 
reversible operation $O_A$ can be implemented in an asymptotically 
thermodynamically reversible manner by using an appropriate switching 
circuit that is operated in a truly, fully adiabatic way, 
transitioning among some discrete set of logic levels.
\end{assertion}

Although we will not here prove Assertion~\ref{ass:adiacr} formally, 
part (a) essentially follows from our earlier observation in 
Theorem~\ref{thm:genrc} that, in deterministic computations, 
conditional reversibility is the correct statement of the 
logical-level requirement for avoiding energy dissipation under 
{\lanprinc}, and therefore it is a necessity for approaching 
thermodynamic reversibility in any deterministic computational process,
and therefore, more specifically, in the operation of adiabatic 
circuits.  

Meanwhile, part (b) follows simply from general constructions showing 
how to implement any desired conditioned reversible operation in an 
asymptotically thermodynamically reversible way using adiabatic 
switching circuits. For example, Fig.~\ref{fig:rCopy-circSeq} 
illustrates how to implement an {\texttt{rCOPY}} operation using a 
simple four-transistor CMOS circuit.  In contrast, implementing 
{\texttt{rCOPY}} by embedding it within an unconditionally-reversible 
{\texttt{cNOT}} would require including an {\texttt{XOR}} capability, 
and would require a much more complicated adiabatic circuit, whose 
operation would itself be composed from numerous more-primitive 
operations (such as adiabatic transformations of individual 
MOSFETs~\cite{SEALeR}) that are themselves only conditionally 				% [16]
reversible.

In contrast, the traditional reversible computing framework of 
unconditionally reversible operations does not exhibit any 
correspondence such as that of Assertion~\ref{ass:adiacr} to any 
natural class of asymptotically physically-reversible hardware that we 
know of.  In particular, the traditional unconditionally-reversible 
framework does not correspond to the class of truly/fully adiabatic 
switching circuits, because there are many such circuits that do not 
in fact perform unconditionally reversible operations, only 
conditionally-reversible ones.

\begin{figure}[!tb]
\centering
\includegraphics[height=1.4in]{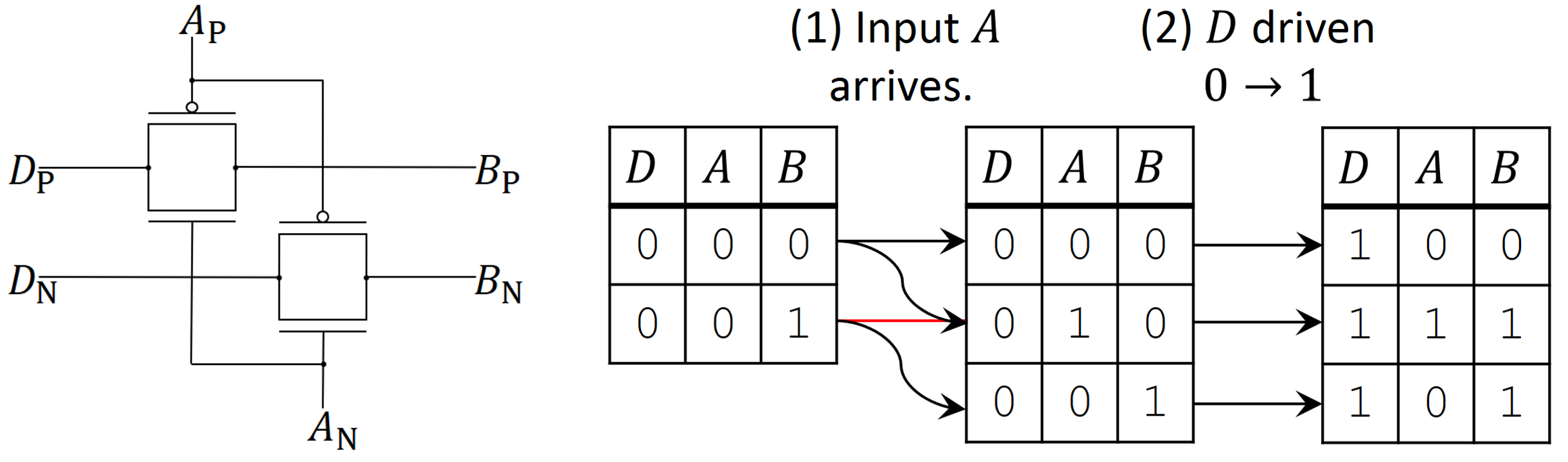}
\caption{(Left) A simple adiabatic CMOS circuit capable of carrying out 
a variant of the {\texttt{rCOPY}} operation.  Here, computational 
states are represented using dual-rail complementary voltage coding, so 
that, for example, a logical state $A=\mathtt{0}$ is represented using 
the voltage assignments $A_\mathrm{P} = V_\mathrm{H}, A_\mathrm{N} = 
V_\mathrm{L}$, where $V_\mathrm{H},V_\mathrm{L}$ are high and low 
voltage levels, respectively.  The logical state $A=\mathtt{1}$ would 
be represented using the opposite voltage assignments.  The two CMOS 
transmission gates shown will thus be turned ON (conducting) only when 
$A=\mathtt{1}$.  In this circuit, $A$ is the logic input, $B$ is the 
output, and $D$ is a driving signal.  (Right) Sequence of operation. 
Assume initially that $D=\mathtt{0}$ and $A=\mathtt{0}$. Normally we 
would also have $B=\mathtt{0}$ initially, but to illustrate the 
conditional reversibility of this circuit, we will also consider the 
case $B=\mathtt{1}$.  In step 1, some external circuit adiabatically 
transforms input $A$ from logic {\texttt{0}} to a newly-computed value 
(\texttt{0} or {\texttt{1}}) to be copied, then in step 2, the drive 
signal $D$ is unconditionally transformed adiabatically from logic 
\texttt{0} to \texttt{1}.  Note that, in the course of this operation 
sequence, if $B$ were {\texttt{1}} initially, then it would be 
dissipatively sourced to $D=\mathtt{0}$ in step 1 if $A=\mathtt{1}$.  
Thus, this particular operation sequence implements a conditioned 
reversible operation $\mathtt{rCOPY}'(A,B\,|\,\overline{AB})$; it is 
reversible as long as we don't try to copy an input value 
$A=\mathtt{1}$ onto an initial state where $B=\mathtt{1}$. The prime 
there after {\texttt{rCOPY}} is denoting the variant semantics, namely 
that in the case $\bar{A}B$, the value $A=\mathtt{0}$ is not copied to 
$B$.}
\label{fig:rCopy-circSeq}
\end{figure}

%---------1---------2---------3---------4---------5---------6---------7
\section{Comparison to Prior Work}
\label{sec:compri}

The concept of conditional reversibility presented here is similar to, 
but distinct from, certain concepts that are already well known in the 
literature on the theory of reversible circuits and languages.

First, the concept of a reversible computation that is only 
semantically correct (for purposes of computing a desired function) 
when a certain precondition on the inputs is satisfied is one that was 
already implicit in Landauer's original paper~\cite{land}, when he 			% [2*]
introduced the operation now known as the Toffoli gate, as a reversible 
operation within which Boolean {\texttt{AND}} may be embedded.  
Implicit in the description of that operation is that it only correctly 
computes {\texttt{AND}} if the program/output bit is initially 0; 
otherwise, it computes some other function (in this case, 
{\texttt{NAND}}).  This is the origin of the concept of 
{\emph{ancilla}} bits, which are required to obey certain pre- and 
post-conditions (typically, being cleared to 0) in order for reversible 
circuits to be composable and still function as intended.  The study of 
the circumstances under which such requirements may be satisfied has 
been extensively developed, {\emph{e.g.}} as in \cite{ricercar}.  			% [17]
However, any circuit composed from Toffoli gates is {\emph{still 
reversible}} even if restoration of its ancillas is violated; it may 
yield nonsensical outputs in that case, when composed together with 
other circuits, but at no point is information erased.  This 
distinguishes ancilla-preservation conditions from our preconditions 
for reversibility, which, when they are unsatisfied, necessarily yield 
actual (physical) irreversibility.

Similarly, the major historical examples of reversible high-level 
programming languages such as Janus~(\cite{lutz,yoko10}), 					% [18-19]
$\mathrm{\Psi}$-Lisp~\cite{baker}, the author's own R 						% [20]
language~\cite{frank}, and RFUN~(\cite{yoko11,ag13}) have 					% [21] [22-23]
invoked various ``preconditions for reversibility'' in the defined 
semantics of many of their language constructs.  But again, that 
concept really has more to do with the ``correctness'' or 
``well-definedness'' of a high-level reversible program, and this 
notion is distinct from the requirements for actual physical 
reversibility during execution.  For example, the R language 
compiler {\cite{frank}} generated PISA assembly code in such a way 			% [21*]
that even if high-level language requirements were violated 
({\emph{e.g.}, in the case of an {\texttt{if}} condition changing its 
truth value during the {\texttt{if}} body), the resulting assembly 
code would still execute reversibly, if nonsensically, on the 
Pendulum processor~\cite{vieri}.  											% [24]

In contrast, the notion of conditional reversibility explored in the 
present document ties directly to Landauer's principle, and to the 
possibility of the physical reversibility of the underlying hardware.  
Note, however, that it does not concern the semantic correctness of the 
computation, or lack thereof, and in general, the necessary 
preconditions for the physical reversibility and correctness of a given 
computation may be orthogonal to each other, as illustrated by the 
example in Fig.~\ref{fig:rCopy-circSeq}.  

%---------1---------2---------3---------4---------5---------6---------7
\section{Conclusion}
\label{sec:concl}

In this paper, we have formally presented the core foundations of a 
general theoretical framework for reversible computing.  We analyzed 
the case of deterministic computational operations in detail, and 
proved that the class of deterministic computations that are not 
required to eject any entropy from the computational state under 
{\lanprinc} is larger than the set of computations composed of the 
unconditionally-reversible operations considered by traditional
reversible computing theory, because it also includes the set of 
conditionally-reversible operations whose preconditions for 
reversibility are satisfied with probability 1 (or asymptotically 
approaching 1, if we only need the entropy ejected to approach 0).  
This is, moreover, the most general possible characterization of the 
set of classical deterministic computations that can be physically 
implemented in an asymptotically thermodynamically reversible way.  

We then demonstrated some applications of the theory by illustrating 
some basic examples of conditioned reversible operations that work by 
transforming an output variable between a predetermined, known value 
and the computed result of the operation.  Such operations can be 
implemented in a very simple way using adiabatic switching circuits, 
whose computational function cannot in general be represented within 
the traditional theory of unconditionally-reversible computing.  This 
substantiates our assertion that the generalized reversible computing 
theory is deserving of significantly greater emphasis than it has so
far received.  

Some promising directions for future work include: (1) Giving further 
examples of useful conditioned reversible operations; (2) illustrating 
detailed physical implementations of devices for performing such 
operations; (3) further extending the development of the new framework 
to address the nondeterministic case, in which operations can be 
non-entropy-ejecting, or even entropy-absorbing, even when none of 
their preconditions for (logical) reversibility are satisfied; (4) 
developing further descriptive frameworks for reversible computing at 
higher levels (\emph{e.g.}, hardware description languages, programming 
languages) building on top of the fundamental conceptual foundations that 
GRC theory provides.

The further study and development of Generalized Reversible Computing 
theory, since it broadens the range of design possibilities for 
reversible computing devices in a clearly delineated, well-founded way, 
will be essential if the computing industry is going to successfully 
transition, over the coming decades, to the point where it is 
dominantly utilizing the reversible computing paradigm.  Due to the 
incontrovertible validity of {\lanprinc}, such a transition will be an 
absolute physical prerequisite in order for the energy efficiency (and 
cost efficiency) of general computing technology 
(that is, beyond the few cases that may be substantially sped up by 
quantum algorithms) to continue growing by indefinitely many orders 
of magnitude.

%{\vfill}

%---------1---------2---------3---------4---------5---------6---------7

\end{document}